\documentclass[a4paper,12pt]{article}
\usepackage[14pt]{extsizes} 
\usepackage[mathscr]{eucal}
\usepackage{amsmath,amsfonts,amssymb,amsthm}
\usepackage{times}
\usepackage{hyperref}

\newtheorem{thm}{Theorem}

\newtheorem{theorem}{\textbf{Theorem}}[section]
 \newtheorem{prop}{Proposition}[section]
 \newtheorem{lemma}[thm]{\textbf{Lemma}}
 \newtheorem{definition}[prop]{Definition}

\renewcommand{\tilde}{\widetilde}
\renewcommand{\hat}{\widehat}

\newcommand{\bref}[1]{\textbf{\ref{#1}}}

\newcommand{\p}[1]{|#1|}
\newcommand{\gh}[1]{\mathrm{gh}(#1)}

\newcommand{\dv}{\mathrm{d_v}}
\newcommand{\dx}{\mathrm{d_X}}

\renewcommand{\d}{\partial}
\renewcommand{\dh}{\mathrm{d_h}}

\newcommand{\cM}{\mathcal{M}}

\renewcommand{\geq}{\,{\geqslant}\,}

\newcommand{\binner}[2]{%
  {\langle}\kern-4.15pt{\langle}#1{,}\,#2{\rangle}\kern-4.15pt{\rangle}}
\newcommand{\commut}[2]{[#1{,}\,#2]}

\newcommand{\pb}[2]{\left\{{}#1{},{}#2{}\right\}}

\newcommand{\half}{\mathchoice{%
    \ffrac{1}{2}}{\frac{1}{2}}{\frac{1}{2}}{\frac{1}{2}}}

\newcommand{\ffrac}[2]{\raisebox{.5pt}%
  {\footnotesize$\displaystyle\frac{#1}{#2}$}\kern1pt}

\newcommand{\dl}[1]{\mathchoice{\ffrac{\d}{\d #1}}{\frac{\d}{\d #1}}{\ffrac{\d}{\d #1}}{\ffrac{\d}{\d #1}}}

\newcommand{\st}[2]{{\overset{#1}{#2}}}

\newcommand{\ddl}[2]{\ffrac{\d #1}{\d #2}}

\newcommand{\Liealg}{\mathfrak} 
\newcommand{\algg}{\Liealg{g}}

\newcommand{\algA}{\mathcal{A}}
\newcommand{\CC}{\mathcal{C}}
\newcommand{\cC}{\mathcal{C}}

\newcommand{\fR}{\mathbb{R}}

\newcommand{\fZ}{\mathbb{Z}}

\newcommand{\bA}{\mathbf{A}}
\newcommand{\bF}{\mathbf{F}}

 \def\cE{\mathcal{E}}
 \def\cG{\mathcal{G}}
 
 \def\cI{\mathcal{I}}
 
\def\cK{\mathcal{K}}
 \def\cL{\mathcal{L}}
 
\def\tr{{\rm Tr}}

\newcommand{\symp}{\omega}
\newcommand{\sect}{\sigma}
\newcommand{\dmsn}{n}
\newcommand{\gtj}{\tilde}
\usepackage[left=3cm,right=2cm,
    top=3cm,bottom=3cm,bindingoffset=0cm]{geometry}

\newcommand{\tgamma}{q}
\usepackage{authblk}
\usepackage{tikz-cd}
\usepackage{comment}
\usepackage{xcolor}

\numberwithin{equation}{section} 

\newcommand{\rchanged}[1]{{#1}}

\usepackage{lipsum}

\newcommand\blfootnote[1]{%
  \begingroup
  \renewcommand\thefootnote{}\footnote{#1}%
  \addtocounter{footnote}{-1}%
  \endgroup
}

\title{Presymplectic minimal models of local gauge theories}

\author[1,2]{Ivan Dneprov }
\author[3,$*$,$\dagger$,$\ddagger$]{~~~~Maxim Grigoriev }
\author[1]{Vyacheslav~~Gritzaenko}
\affil[1]{\textsl{ Lebedev Physical Institute, 
  Leninsky ave. 53, 119991 Moscow, Russia \vspace{5pt}}} 
\affil[2]{\textsl{ Institute for Theoretical and Mathematical Physics,\protect\\
  Lomonosov Moscow State University, 119991 Moscow, Russia  \vspace{5pt}}}
\affil[3]{\textsl{ Service de Physique de l'Univers, Champs et Gravitation, \protect\\ Universit\'e de Mons, 20 place du Parc, 7000 Mons, 
Belgium \vspace{-20pt}}}
\date{}    

\begin{document}

\maketitle

\begin{abstract}
We elaborate on the recently proposed notion of a weak presymplectic gauge PDE. It is a $\mathbb{Z}$-graded bundle over the space-time manifold, equipped with a degree $1$ vector field and a compatible graded presymplectic structure. This geometrical data naturally defines a Lagrangian gauge field theory. Moreover, it encodes not only the Lagrangian of the theory but also its full-scale Batalin-Vilkovisky (BV) formulation. In particular, the respective field-antifield space arises as a symplectic quotient of the super-jet bundle of the initial fiber bundle. A remarkable property of this approach is that among the variety of presymplectic gauge PDEs encoding a given gauge theory we can pick a minimal one that usually turns out to be finite-dimensional, and unique in a certain sense. The approach can be considered as an extension of the familiar AKSZ construction to not necessarily topological and diffeomorphism-invariant theories. We present a variety of examples including $p$-forms, chiral Yang-Mills theory, Holst gravity, and conformal gravity. We also explain the explicit relation to the non-BV-BRST version of the formalism, which happens to be closely related to the covariant phase space and the multisymplectic approaches.
\end{abstract}
\begin{flushright}
\textit{In memory of Igor Anatolievich Batalin}
\end{flushright}
\phantom{a}
\blfootnote{${}^{*}$ {Corresponding Author: grigoriev.max@gmail.com}}
\blfootnote{${}^{\dagger}$ Supported by the ULYSSE Incentive
Grant for Mobility in Scientific Research [MISU] F.6003.24, F.R.S.-FNRS, Belgium.}
\blfootnote{${}^{\ddagger}$ On leave of absence from Lebedev Physical Institute and Institute for Theoretical and Mathematical Physics, Lomonosov MSU.}
\tableofcontents

\section{Introduction}

The Batalin-Vilkovisky (BV)~\cite{Batalin:1981jr,Batalin:1983wj} formalism is now considered to be the most fundamental mathematical setup for studying gauge field theories. Besides its original use as a quantization tool, it turned out to be extremely fruitful in studying symmetries, consistent interactions, and renormalization, see e.g.~\cite{Voronov:1982ur,Barnich:1993vg,Barnich:1994ve,Piguet:1995er,Barnich:2000zw,costello2011renormalization}, as well as in constructing new models as BV systems from the very start. The well-known examples of the latter are the String Field Theory~\cite{Thorn:1986qj,Bochicchio:1986zj,Zwiebach:1993ie} and topological field theories.  In the case of topological theories, it is the so-called AKSZ construction~\cite{Alexandrov:1995kv} which encodes the BV formulation of a topological model in terms of the space-time manifold and the finite-dimensional super-geometrical  object -- $QP$-manifold, i.e. a graded manifold equipped with a homological vector field and a compatible graded symplectic structure.

The AKSZ construction has been successfully employed in describing various topological models, see e.g.~\cite{Cattaneo:1999fm,Batalin:2001fc,Cattaneo:2001ys,Roytenberg:2002nu,Bonechi:2009kx,Barnich:2009jy,Bonavolonta:2013mza,Ikeda:2012pv}. Moreover, it has been shown~\cite{Grigoriev:1999qz} to relate the Lagrangian BV and the Hamiltonian BFV~\cite{Fradkin:1975cq,Batalin:1977pb,Fradkin:1977xi} formulations of constrained Hamiltonian systems. In this sense, the AKSZ construction comprises both the BV and BFV formulations of the underlying gauge system~\cite{Barnich:2003wj,Grigoriev:2010ic,Grigoriev:2012xg} (see also~\cite{Cattaneo:2012qu,Cattaneo:2015vsa} for a related approach to unification of BV and BFV).

The AKSZ construction, as nice as it is, does not directly apply to non-topological field theories\footnote{Note that reparameterization invariant mechanical systems are of AKSZ type~\cite{Grigoriev:1999qz,Barnich:2003wj}.} which are of the main interest from the physics perspective. Of course, this is the case only if one insists on manifest locality and keeps the target $QP$-manifold finite-dimensional. 

A possible way out is to allow for infinite-dimensional target space. If we limit ourselves to the non-Lagrangian (equations of motion)  version~\cite{Barnich:2006hbb} of the AKSZ construction, the relevant generalization is known~\cite{Barnich:2010sw} (see also~\cite{Barnich:2004cr,Grigoriev:2006tt} for the simplified versions and \cite{Grigoriev:2019ojp} for a modern geometrical description) and amounts to taking as target space  the jet-bundle of the BV formulation of the theory in question seen as a $Q$-manifold with $Q=s+\dh$, i.e. the total BRST differential~\cite{Barnich:2010sw,Grigoriev:2012xg}. This naturally leads to a concept of a gauge PDE (gPDE for short)~\cite{Grigoriev:2019ojp}) which encodes a local gauge theory in the geometric data of a $Q$-bundle~\cite{Kotov:2007nr}, i.e. an AKSZ-like fiber bundle equipped with a compatible homological vector field. Let us also mention a somewhat related unfolded formalism of higher spin gauge theories~\cite{Vasiliev:1988xc,Vasiliev:2005zu}.

If one is after a full-scale AKSZ formulation of a given Lagrangian gauge theory, the situation is more subtle. The required  modification is much less trivial and was put forward in~\cite{Grigoriev:2010ic,Grigoriev:2012xg}. The crucial ingredient of the construction is the descent-completion $\cL$ of the Lagrangian $\cL_n$ to a cocycle of the total BRST differential $s_0+\dh$ by adding a lower-degree horizontal forms $\cL_{n-1},\cL_{n-2},\ldots$.

An interesting alternative, initially proposed in~\cite{Alkalaev:2013hta,Grigoriev:2016wmk}, is to consider instead of a symplectic structure a possibly degenerate presymplectic one. It turns out that in this case an AKSZ-like model with the finite-dimensional target-space can describe non-topological systems. \rchanged{This phenomenon}
can be understood as follows: given a Lagrangian BV system one can for the moment forget about the Lagrangian and concentrate on the equations of motion and gauge symmetries in order to equivalently reformulate it as a gauge PDE.  Furthermore, among the equivalent gauge PDE formulations of the system one can usually identify the minimal one. This can be obtained by eliminating the maximal amount of contractible pairs for the total BRST differential. In the next step one finds that the Lagrangian and the BV symplectic structure can be encoded in a compatible graded presymplectic structure defined on the gauge PDE. This presymplectic structure is (an equivalent of) the descent-completion of the BV-symplectic structure to a cocycle of the total BRST differential $s+\dh$. The crucial observation is that  the AKSZ-like action determined by this data gives an equivalent Lagrangian formulation of the system and moreover involves only a finite number of fields because the presymplectic structure is a local form.  Furthermore, it turns out that not only the Lagrangian but also a full-scale BV formulation can be reconstructed from this data by taking a symplectic quotient of the space of AKSZ fields~\cite{Grigoriev:2020xec,Dneprov:2022jyn} (see also~\cite{Grigoriev:2016wmk}). 

An additional observation is that all the coordinates along the kernel of the presymplectic structure are passive in the construction, e.g. the Lagrangian does not depend on them. This allows one to take the quotient with respect to the corresponding  subdistribution of the kernel distribution, resulting in a finite-dimensional bundle equipped with a compatible presymplectic structure and degree $1$ vector field $Q$. The price to be paid is that $Q$ is in general not nilpotent but satisfies a presymplectic analog of the BV master equation, which, eventually, ensures that the vector field induced on the symplectic quotient is again nilpotent and the full-scale BV formulation is recovered therein.  This leads to the concept of weak presymplectic gPDEs~\cite{Grigoriev:2022zlq}. 

In this work we perform a systematic study of weak presymplectic gPDEs. More specifically, we propose a more useful set of the basic axioms and give a detailed proof that, at least locally, a weak gPDE defines a complete BV formulation. Because weak gPDEs arise as certain quotients  of presymplectic gPDEs, of a particular interest are weak presymplectic gPDEs associated to minimal gPDEs. These are gPDEs where the maximal amount of contractible pairs have been eliminated. In the formal case, a minimal gPDE corresponds to the minimal models of the associated  $L_\infty$ (recall that formal $Q$-manifolds are $1:1$ with $L_\infty$ algebras) and hence should be unique, under some technical conditions. It follows, that weak gPDEs corresponding to minimal models are equally distinguished provided the entire regular part of the kernel distribution was factored out. When applied to the standard examples of topological models, such as the Chern-Simons theory, the procedure described above results in the usual AKSZ formulations. In this sense,weak presymplectic gPDEs can be thought of as a generaliztaions of AKSZ models to the case of not necessarily topological and diffeomorphism-invariant local gauge theories.

In addition to general statements we present a variety of examples of weak presymplectic gPDEs that describe some known local gauge theories. These include chiral YM theory~\cite{Chalmers:1996rq}, conformal gravity, Friedman-Townsend theory, Holst gravity~\cite{Holst:1995pc}, and Plebansky gravity~\cite{Plebanski:1977zz}. We also discuss $p$-form gauge theory as an illustration of how our machinery encodes a rather extended BV field-antifield space of the model into a very concise super-geometrical object.

The paper is organized as follows: Section~\bref{sec:prelim} contains a recap of graded geometrical approaches to BV and related constructions. In Section~\bref{sec:main} we introduce weak presymplectic gPDEs and prove that they define a BV theory. Section\bref{sec:examples} consists of a variety of examples of weak presymplectic gPDEs.

\section{Preliminaries}
\label{sec:prelim}

\subsection{Local BV systems}

BV formulation of local gauge theories is usually defined in terms of underlying graded fiber bundles and their associated jet-bundles. Here we give a very brief exposition with an accent on geometrical structures. The standard exposition can be found in e.g.~\cite{Barnich:2000zw}. 

Let $\cE$ be a graded fiber bundle over a real spacetime manifold $X$. Its sections (strictly speaking supersections ) are to be identified as BV fields and antifields with the degree as the ghost degree.
\begin{definition}
\label{def:local-BV}
A local BV  system  with the underlying fiber bundle $\cE\to X$,  \rchanged{$\dim(X)=\dmsn$},
is determined by the following data:\\
(i) a degree-1, evolutionary vector field $s$ defined on $J^\infty(\cE)$ and satisfying $s^2=0$\\
(ii) an $(n,2)$-form $\st{n}{\omega}\in \bigwedge^{(n,2)}(J^\infty(\cE))$ of ghost degree $-1$, which is a pullback to $J^\infty(\cE)$ of a closed $n+2$ form $\omega^\cE$ on $\cE$, 
such that
\footnote{\rchanged{Here and in what follows $L_W \equiv i_W d +(-1)^{\p{W}} di_W$ denotes the Lie derivative along the vector field $W$.}}
\begin{equation}
\label{eq:s-inv}
L_s\st{n}{\omega}+\dh (\ldots)=0\,, 
\end{equation}
In addition, $\omega^\cE$ is required not to have zero-vectors.  
\end{definition}
Note that in the adapted local coordinates $x^a,\psi^A$, where $x^a$ are base-space coordinates, $\omega^\cE$ satisfying the above conditions can be represented as
\begin{equation}
\label{cEform}
    \omega^\cE = (dx)^n \omega_{AB}(\psi,x)d\psi^A d\psi^B\,,\qquad (dx)^n=\frac{1}{n!}\epsilon_{a_1\ldots a_n} dx^{a_1}\ldots dx^{a_n}
\end{equation}
with $\omega_{AB}$ invertible. In this form it is clear why it defines a graded Poisson bracket.

Let us recall how this data is related to a "more conventional" definition of BV in terms of master action and antibracket.  To simplify the discussion we assume that the de Rham cohomology of $X$ is empty in positive degree. Using $L_s \st{n}{\omega} = (i_s \dv - \dv i_s)\st{n}{\omega} = -\dv i_s\st{n}{\omega}$ and equation \eqref{eq:s-inv} one has $\dv i_s\st{n}{\omega}=\dh(\ldots)$. Using that $\dv$ mod $\dh$ is empty in this degree, at least locally, one finds that there exists $\st{n}{\cL}$ such that 
\begin{equation} \label{i_s_form}
    i_s\st{n}{\omega} + \dv \st{n}{\cL} = \dh(\ldots) \,. 
\end{equation}
In other words $s$ is a Hamiltonian vector field and $\st{n}{\cL}$ is its Hamiltonian. The horizontal $n$-form $\st{n}{\cL}$ is usually called the BV Lagrangian.

Analogous considerations apply to any Hamiltonians  vector field, i.e. satisfying $L_V \st{n}{\omega}=\dh(\ldots)$, so that the analog of \eqref{i_s_form} defines the associated Hamiltonians $H_V \in \bigwedge^{(n,0)}(J^\infty(\cE))$. This defines $H_V$ only modulo the $\dh$-exact contributions.  The odd Poisson bracket of a pair of such Hamiltonians can be defined as
\begin{equation}
   \pb{H_{V_1}}{H_{V_2}}=i_{V_1}i_{V_2}\st{n}{\omega} \,,
\end{equation}
and is well-defined on equivalence classes, i.e. on $H^{(n,0)}(\dh)$. One can check that the above bracket makes $H^{(n,0)}(J^\infty(\cE))$ a graded Lie algebra. Applying these considerations to $\st{n}{\cL}$ one finds, 
\begin{equation}
\pb{\st{n}{\cL}}{\st{n}{\cL}}=i_si_s \st{n}{\omega}=\dh(\cdot)\,.
\end{equation}
i.e. a classical BV master equations written in terms of BV Lagrangian density.

Alternatively, restricting ourselves to objects of compact support or assuming a suitable asymptotic/boundary conditions one can extend this bracket to genuine functionals of fields. In this form it can be identified with the standard BV antibracket. Of course these considerations are completely standard and apply to generic symplectic/Poisson structures, not necessarily BV ones. For more details see e.g.~\cite{Dickey:1991xa,Barnich:1997ed,Sharapov:2016sgx}.

The BV action functional can be written as:
\begin{equation}
S_{BV}(\hat\sigma)=\int_X \hat\sigma_{pr}^*(\st{n}{\cL})\,, 
\end{equation}
where $\hat\sigma$ is a supersection of $\cE$ and $\hat\sigma_{pr}$ denotes its infinite prolongation. As usual, the underlying classical action is obtained by setting fields of non-vanishing degree to zero, or, more geometrically, restricting $S_{BV}$ to sections. 

The field-theoretical interpretation of a local BV system can be given entirely in terms of the geometry of $\cE$ without resorting to the BV action. Indeed, a section $\sigma : X \to \cE$ is a solution if its prolongation $\sigma_{pr}$ satisfies $\sigma_{pr}^* \circ s=0$, i.e. that $\sigma_{pr}$ belongs to the zero locus of $s$. Gauge transformations are defined in terms of a vertical evolutionary vector field $\xi$ of ghost degree $-1$, namely the infinitesimal gauge transformation associated to $\xi$ is an evolutionary vector field $[s,\xi]$.
\begin{equation}
    \delta \sigma^*_{pr} = \sigma^*_{pr} \circ [s,\xi] \,.
\end{equation}
In a similar way one defines the (higher order) gauge for gauge transformations.

\subsection{Q-manifolds and gauge PDEs}

Throughout the paper we extensively employ the language of graded geometry, $Q$-manifolds (also known as dg-manifolds), and their associated fiber bundles. By definition, a $Q$-manifold is a pair $(\cM,Q)$, where \rchanged{$\cM$}
is a $\fZ$-graded supermanifold and $Q$ an odd vector field of ghost degree 1
satisfying $Q^2=0$. In all our examples the Grassmann degree is induced by $\fZ$-degree, i.e. $\p{f}=\gh{f}\, \text{mod}\, 2$, so for simplicity we do not write it explicitly. This can always be reinstated. 

\rchanged{An important example of a $Q$-manifold is $T[1]X$, the shifted tangent bundle of a real manifold $X$. If $x^a$ are local coordinates on $X$
the induced local coordinate system on $T[1]X$ is given by $x^a, \theta^a$ with $\gh{x^a}=0$, $\gh{\theta^a}=1$. Algebra of functions on $T[1]X$ can be identified with the exterior algebra of the underlying manifold $X$. Under this identification the de Rham differential on $X$ corresponds to a homological vector field $\dx$ on $T[1]X$. In terms of the above local coordinates  $\dx$ reads as $\dx=\theta^a \ddl{}{x^a}$.}

A fiber bundle in the category of $Q$-manifolds is known as a $Q$-bundle~\cite{Kotov:2007nr}. Both total space and the base are $Q$-manifolds in this case and the $Q$-structures are required to be compatible with the bundle projection, i.e. if $\pi:(\rchanged{\cM},Q)\to (B,q)$ is a projection one has: $\pi^*\circ q=Q\circ \pi^*$.

A rather far-going generalization of BV formulation of local gauge theories at the level of equations of motion can be cast into the geometric data of a $Q$-bundle. 
The following definition is a geometrical and refined version of the formalism put forward in~\cite{Barnich:2010sw,Grigoriev:2012xg}:
\begin{definition}\cite{Grigoriev:2019ojp}
    (i) Gauge PDE $(E,Q,T[1]X)$ is a $Q$ bundle $\pi:(E,Q)\rightarrow(T[1]X,\dx)$, where $X$ is a real manifold (independent variables). In addition it is assumed that $(E,Q,T[1]X)$ is locally equivalent to a nonnegatively graded $Q$-bundle.\\
    (ii) Sections of $E$ are interpreted as field configurations. Moreover, $\sigma: T[1]X\rightarrow E$ is a solution to $(E,Q,T[1]X)$ if 
\begin{align}\label{solutions}
        \dx\circ\sigma^{*}=\sigma^{*}\circ Q\,.
\end{align}
    (iii) Infinitesimal gauge transformations of $\sigma$ are defined as    
\begin{align}\label{gaugetransf}
        \delta\sigma^{*}={\sigma}^{*}\circ \commut{Q}{Y}\,,
    \end{align}
    where $Y$ is a vertical vector field of ghost degree $-1$ which is to be understood as a gauge parameter.  In a similar way one defines (higher order) gauge for gauge symmetries.
\end{definition}
It is easy to check that transformation~\eqref{gaugetransf} is indeed an infinitesimal symmetry of the equations of motion~\eqref{solutions}.
Note that the gauge transformations can be also defined via
\begin{equation}
\label{aksz-gs}
    \delta\sigma^*=d_X \circ \xi_\sigma^*+\xi_\sigma^* \circ Q\,,
\end{equation}
where $\xi_\sigma^*:\CC^\infty(E)\to \CC^\infty(T[1]X)$ is a gauge parameter map, which has degree $-1$ and satisfies
\begin{equation}
\label{gp-cond}
    \xi_\sigma^*(fg)=(\xi_\sigma^*(f)\sigma^*(g)+(-1)^{\p{f}}\sigma^*(f)\xi_\sigma^*(g)\,,
\end{equation}
as well as $\xi^*_\sigma(\pi^*\alpha)=0$ for any $\alpha \in \CC^\infty(T[1]X)$. It is easy to see that taking $\xi_\sigma=\sigma^*\circ Y$
one indeed recovers \eqref{gaugetransf} if $\sigma$ is a solution. Recall that only the gauge transformations of solutions have invariant meaning.

\subsection{AKSZ models}

The celebrated AKSZ construction~\cite{Alexandrov:1995kv} has been initially found as an elegant geometrical object which encodes a full-scale Lagrangian BV formulation, including BV master action and BV antibracket. In the terminology of the previous subsection, the AKSZ construction is a (usually finite-dimensional) gauge PDE which is globally trivial as a $Q$-bundle, i.e. $(E,Q)=(T[1]X,\dx)\times (F,q)$, whose fiber $F$ is in addition equipped with a compatible symplectic structure $\omega^F$ of ghost degree $n-1$, where $n=\dim(X)$.
The compatibility condition reads 
\begin{equation}
    L_Q\omega=0\,.
\end{equation}
Here $\omega$ is the symplectic structure on the total space $E$ defined by the target space symplectic structure $\omega^F$.
Note that in the adapted coordinates $Q=\dx+q$, i.e. a product $Q$-structure. A natural generalization is to consider source manifolds which are more general than $(T[1]X,\dx)$, but we refrain from doing this here.

The fields and antifields of the theory are supersections of $E$. Let $\psi^A$ denote local coordinates on $F$ and $x^a, \theta^a$ on $T[1]X$. Their pullbacks to  $F \times T[1]X$
define coordinates on $E$, which by some abuse of conventions we keep denoting as $\psi^A,x^a, \theta^a$. Coordinates on the space of supersections are then introduced as:
\begin{equation}
    \Hat{\sigma}^*(\psi^A) = \st{0}{\psi}{}^A(x) + \st{1}{\psi}{}^A_a(x)\theta^a + ...
\end{equation}
Note that $\gh{\st{k}{\psi}{}_{...}^A(x)}=\gh{\psi^A}-k$.

The above geometrical data defines a BV master-action functional on the space of super-sections. Namely, 
the BV action of the theory is then given by 
\begin{equation} \label{AKSZ-action}
    S[\hat\sigma] = \int_{T[1]X} \Hat{\sigma}^*(\chi)(\dx) + \Hat{\sigma}^*(\cL)\,,
\end{equation}
where $\chi$ denotes a symplectic potential, i.e. $\omega=d\chi$,  $\cL$ is defined through $i_Q \omega+d\cL=0$, and \rchanged{$\Hat{\sigma}^*(\chi)(\dx)$ denotes the evaluation of the 1-form $\Hat{\sigma}^*(\chi)$ on the vector field $\dx= \theta^a \ddl{}{x^a}$ i.e. $\Hat{\sigma}^*(\chi)(\dx)=i_{\dx}\Hat{\sigma}^*(\chi)$.}

Moreover, the symplectic structure on $F$ determines in a standard way a symplectic structure $\bar\omega$ on the space of super-sections. If $\hat\sigma$ is a given supersection (a point in the space of super-sections) and $\delta_1\hat\sigma$ and $\delta_2 \hat\sigma$ are two tangent vectors at $\hat\sigma$ then the value of $\bar\omega$ evaluated at $\hat\sigma$ on these tangent vectors is given by:
\begin{equation}
\label{indomega}
\bar\omega_{\hat\sigma}(\delta_1\hat\sigma,\delta_2\hat\sigma)=\int_{T[1]X} \omega_{\hat\sigma}
(\delta_1\hat\sigma,\delta_2\hat\sigma) \,,
\end{equation}
where $\hat\sigma$ is seen as a map from $T[1]X$ to $E$ while $\delta_{1,2}\hat\sigma$ as maps
from $T[1]X$ to $T_{\hat\sigma}E$ so that the integrand is indeed a function on $T[1]X$.

It is instructive to give a concise representation for $\bar\omega$ in terms of coordinates. To this end let us treat $\theta^a$ as auxiliary coordinates and introduce generating functions for coordinates on the space of supersections as:
\begin{equation}
\Psi(x|\theta)=\st{0}{\psi}{}^A(x) + \st{1}{\psi}{}^A_a(x)\theta^a + ...
\end{equation}
and similarly for the basis differentials $\delta\st{k}{\psi}{}^A(x)$. Then 
\begin{equation}
    \Bar{\omega} = \int d^nx d^n\theta \omega_{AB}(\Psi(x|\theta))\delta\Psi^A(x|\theta)\delta \Psi^B(x|\theta)\,.
\end{equation}

Finally, let us spell out the explicit form of the equations of motion derived from the action \eqref{AKSZ-action}:
\begin{equation}
\label{EoM_AKSZ1}
    \sigma^*(\omega_{AB})(d_X\sigma^*(\psi^A) -\sigma^*(Q \psi^A)) = 0\,.
\end{equation}
Note that for a non-degenerate $\omega_{AB}$ these equations are equivalent to the condition that $\sigma$ is a $Q$ morphism, i.e. solutions to~\eqref{EoM_AKSZ1} also solve the respective gauge PDE $(E,Q, T[1]X)$. Note also that for the value of index $A$ such that $\gh{\psi^A}=-1$ equations $d_X\sigma^*(\psi^A) -\sigma^*(Q \psi^A) = 0$ takes the form $\sigma^*(Q\psi^A)=0$, i.e. is a constraint~\cite{Barnich:2006hbb}.

\subsection{Presymplectic gauge PDEs}

Although AKSZ sigma models give a very compact and elegant way to describe BV formulation of interesting gauge field theories in its traditional form the approach is limited to topological (=no local degrees of freedom) and diffeomorphism invariant systems. 

The analog of AKSZ construction in the general case can be arrived at as follows: 
given a local BV system one can for the moment forget about the Lagrangian and concentrate on the equations of motion and gauge symmetries. It is known~\cite{Barnich:2010sw} that, at least locally, it can be equivalently represented as a gauge PDE. Furthermore, it turns out that the Lagrangian and the BV symplectic structure can be encoded in the graded \rchanged{presymplectic}
structure defined on the gPDE. This procedure is formalised in the following definition:
\begin{definition}
    A presymplectic gauge PDE is a gauge PDE $(E,Q,T[1]X)$ equipped with a presymplectic structure $\omega$ such that 
    \begin{equation}
    \label{gPDE}
        d\omega =  0, \qquad L_Q\omega \in \cI\,,
    \end{equation}
where \rchanged{ $\cI \subset \bigwedge^{\bullet}(E)$ denotes the ideal in $\bigwedge^{\bullet}(E)$ generated by differential forms of the form $\pi^*(\alpha)$.}
\end{definition}
This data defines a local BV system, provided $\omega$ satisfies certain regularity condition.
In particular, it defines an action functional on the space of sections of $E$, which generalizes the AKSZ action~\eqref{AKSZ-action},
\begin{equation}
\label{pgPDE-action}
    S[\sigma]=\int_{T[1]X}(\sigma^*\chi)(\dx)+\sigma^*(\cL)\,.
\end{equation}
where $\chi$ is a presymplectic potential $\omega=d\chi$, $\cL$ is defined through $i_Q \omega+d\cL \in \cI$ and \rchanged{$(\sigma^*\chi)(\dx)$ is the evaluation of ${\sigma}^*(\chi)$ on the vector field $d_X = \theta^a \ddl{}{x^a}$.}
It is  invariant under the algebraic gauge transformations generated by the kernel distribution of $\omega$. 
Moreover, $\omega$ is typically local (involves only finite number of coordinates) and hence the above action depends on a finite number of fields only. The extension of the action \eqref{pgPDE-action} to supersections determines the BV action of the theory. The respective BV field-antifield space can be obtained as a symplectic quotient of the space of supersections~\cite{Grigoriev:2020xec,Dneprov:2022jyn,Grigoriev:2022zlq}.

In this approach all the information about the system can be encoded in the geometry of a finite-dimensional quotient of the starting point presymplectic gPDE. Indeed, among possible equivalent gPDE representations of a given system one can pick a minimal one (in the formal setup it is simply the one associated to the minimal model of the respective $L_\infty$ algebra which is known to be unique up to $L_\infty$ isomorphism). Furthermore, disregarding the coordinates along the kernel distribution of the presymplectic structure leads to a finite-dimensional geometrical object which knows everything about the local gauge system in question but at the same time is minimal and unique in a certain sense. We call such objects presymplectic minimal models. More generally, even if one doesn't require them to be minimal, objects of this sort are still useful and can be regarded as the weak versions of presymplectic gPDEs. In this work we perform a detailed study of such objects and the way they encode gauge theories.

\subsection{Local forms in terms of graded geometry}
\label{sec:forms-gradgeom}

For our further considerations it is important to clarify the relation between graded geometry language to describe vertical/horizontal forms on fiber bundles and the standard language. 
\begin{definition}\label{def:vert-forms}
    \rchanged{The algebra of
    vertical forms} on a fiber bundle $E\xrightarrow{\pi}M$ is the quotient $\bigwedge^\bullet(E)/\cI$ where $\cI \subset \bigwedge^\bullet(E)$ is the ideal generated by the forms $\pi^*\alpha, \alpha \in \bigwedge^{k > 0}(M)$.
\end{definition}
Let now $\gtj{E}\to X$ be a fiber bundle equipped with a flat connection so that we have a decomposition of the de Rham differential on $\gtj{E}$ into the horizontal and the vertical parts $d=\dh+\dv$. Let also $E\to T[1]X$ be a lift of $\gtj{E}$ to $T[1]X$, i.e. $E=\Pi^* \gtj{E}$, where $\Pi$ is a canonical projection $\Pi: T[1]X\to X$. In plain terms this means that we extend $\gtj{E}$ by the anti-commuting coordinates $\theta^a$ which are basis differentials $dx^a$ understood as auxiliary coordinates.

The algebra of vertical forms on $E$ (in the sense of \bref{def:vert-forms}) is a differential algebra with the differential induced by de Rham: $d^\prime [a]=[da]$,
where $[a]$ denotes an equivalence class of $a \in \bigwedge^\bullet(E)$.
This algebra is isomorphic to the algebra of all forms on $\gtj{E}$ equipped with the vertical differential. To see this let us first define a map $\Upsilon:\bigwedge^\bullet(E) \to  \bigwedge^\bullet(\gtj{E})$. If $x^a,\psi^A$ is an adapted coordinate system on $\gtj{E}$ and $x^a, \theta^a, \psi^A$ is the induced coordinate system on $E$, the map $\Upsilon$ is determined by: 
\begin{equation}
\label{isom}
\begin{gathered}
    \Upsilon(d\psi^A) = \dv\psi^A \,, \\
    \Upsilon(\theta^a) = d x^a \,,  \\
    \Upsilon(dx^a) = 0 = \Upsilon(d\theta^a)\,,  \\
    \Upsilon(f(x,\psi)) = f(x,\psi)\,.
\end{gathered}
\end{equation}
Note that $\Upsilon$ employs an additional structure on $\gtj{E}$ --  the flat connection. 
It is easy to see that $\Upsilon(\cI)=0$ and hence $\Upsilon$ is well defined on the equivalence classes modulo $\cI$. Moreover,
the induced map is an isomorphism between the differential algebras $(\bigwedge^\bullet(E)/\cI,d^\prime)$ and $(\bigwedge^{(\bullet,\bullet)}(\gtj{E}),\dv)$. 

The above isomorphism allows one to define analogs of $\dv$ and $\dh$ in terms of the algebra of vertical forms $(\bigwedge^\bullet(E)/\cI)$. In particular, $\dv$ is represented by $d^\prime$ while for $\dh$ we have:
\begin{equation}
\label{isom2}
\dh \Upsilon([a])=\Upsilon([L_D a])\,, \qquad D=\theta^a D_a\,,
\end{equation}
where $D_a$ are components of the covariant derivative, i.e. $D_a$ is a horizontal lift of $\dl{x^a}$ and $\dh=dx^a D_a$. It is easy to check that the expression in the right hand side does not depend on the choice of the representative.

Finally, note that the construction can be generalized to the case where the connection $D$ is not flat. This is achieved by taking $\Upsilon d\psi^A=d\psi^A-dx^aD_a\psi^A$. The map is still coordinate independent but the induced analog of $\dv$  is not nilpotent.

\section{Weak presymplectic gPDEs}
\label{sec:main}

In this section we define our main objects and show how they encode a full-scale local BV system.

\subsection{Symplectic quotients of pre-$Q$ manifolds}

The following simple statement explains the idea of the construction:
\begin{lemma}
Let $(M,\omega)$ be a presymplectic supermanifold equipped with an odd vector field $Q$ satisfying:
\begin{equation}
        \omega(Q,Q) = 0, \qquad L_Q \omega = 0\,.
\end{equation}
Let $\algA$ be a subalgebra of functions annihilated by the kernel distribution of $\omega$. Then $Q$ preserves $\algA$ and moreover $Q^2f=0$  for any $f \in \algA$.
\end{lemma}
If $\omega$ is regular $\algA$ is just an algebra of functions on the symplectic quotient, at least locally. This gives a useful generating construction for symplectic $Q$-manifolds.  The statement and the proof were in~\cite{Grigoriev:2022zlq} thought it was probably known before.
For completeness  we sketch the proof here as well. Let $\cK$ be a kernel distribution of $\omega$, i.e. $i_\cK\omega=0$. It is involutive thanks to $d\omega=0$. For any $f \in \algA$ and $Z \in \cK$ one has $ZQf=\commut{Z}{Q}f=0$ because $i_{\commut{Q}{Z}}\omega=L_Q i_Z\omega=0$ so that $Q$ preserves $\algA$.  Finally, $Q^2f=0$ for all $f\in \algA$ because
\begin{equation}
\begin{gathered}
    d i_Q i_Q \omega = i_Q d i_Q \omega - L_Q i_Q \omega = i_Q i_Q d \omega - 2i_Q L_Q \omega -i_{[Q,Q]}\omega = 0\,,
\end{gathered}
\end{equation}
so that $Q^2 \in \cK$.

\subsection{Presymplectic local BV systems}
In what follows we need the following statement, whose non-graded version is known, see e.g.~\cite{Gaset:2022ema}:
\begin{lemma} \label{prop:bundle_quotient}
    Let $E \xrightarrow{\pi} X$ be a graded bundle equipped with $\omega \in \bigwedge^{n+2}(E)$, a closed form, such that its kernel $\cK$ has constant rank and is $\pi$ vertical, i.e. $i_K\omega = 0 \rightarrow \pi_* K = 0$. Then, at least locally, $E/\cK$ is also a bundle over $X$ equipped with a non-degenerate form $\gtj{\omega}$.
\end{lemma}

The definition of a local BV system has a straightforward generalisation to the case 
where the symplectic structure is replaced by the presymplectic one:
\begin{definition}\label{def:p-BV}
A presymplectic local BV  system  with the underlying fiber bundle $\cE\to X$ is determined by the following data:\\
(i) a degree-1, evolutionary vertical vector field $s$ defined on $J^\infty(\cE)$
\\
(ii) an $(n,2)$-form $\st{n}{\omega}\in \bigwedge^{(n,2)}(J^\infty(\cE))$ of ghost degree $-1$, which is a pullback to $J^\infty(\cE)$ of a closed $n+2$ form $\omega^\cE$ on $\cE$, such that
\begin{equation}
\label{eq:s-inv-2}
i_s \st{n}{\omega}+\dv \st{n}{\cL}=\dh(\cdot)\,,
\end{equation}
for some $\st{n}{\cL} \in \bigwedge^{(n,0)}(J^\infty(\cE))$. 
In addition, the kernel distribution $\cK$ of $\omega^\cE$ is required to be vertical, i.e. $\cK \subset V\cE \subset T\cE$.\\
(iii) 
\begin{equation}
\label{p-master}
\half i_s i_s\st{n}{\omega}=\dh(\cdot)\,.
\end{equation}
\end{definition}
Note that  compared  to the Definition~\bref{def:local-BV}, the nilpotency condition is replaced by~\eqref{p-master}. At the same time, \eqref{eq:s-inv-2} is equivalent to $L_s\st{n}{\omega}+\dh (\ldots)=0$, at least locally. In the symplectic case these definitions are equivalent.

Note the following statement:
\begin{lemma}\label{prop:nilp-mod-kernel}
In the setting of Definition~\bref{def:p-BV}, $s^2\equiv\half \commut{s}{s}$ belongs to the kernel of $\st{n}{\omega}$, i.e. $i_{s^2}\st{n}{\omega}=0$. 
\end{lemma}
\begin{proof}
Let $I$ denote  the interior Euler projector defined through $I^2=I$ and $I(\dh a)=0$ for any $a\in \bigwedge^{(n-1,s)}(J^\infty(\cE)$, $s\geq 1$ see e.g.~\cite{Anderson1991,Andersonbook} for further details. Applying $I\dv$ to the equation \eqref{p-master} one gets $I(\dv i_s i_s \st{n}{\omega}) = I(\dv \dh(\cdot)) = I\dh(\cdot) = 0$. On the other hand rewriting the l.h.s. one gets
\begin{equation}
    I(-2i_s L_s \st{n}{\omega} - i_{[s,s]}\st{n}{\omega}) =  0\,.
\end{equation}
Using \eqref{eq:s-inv-2} the first term can be rewritten as:
\begin{equation}
    I(-2i_sL_s \st{n}{\omega}) = I(2i_s\dv\dh(\cdot)) = 
    I(\dh(i_s\dv(\cdot))) = 0\,,
\end{equation}
giving $I i_{[s,s]}\st{n}{\omega} = 0$. Because $\st{n}{\omega}$ is a pullback of a form from $\cE$ it does not depend on higher jets and therefore the Euler projector (see \cite{Anderson1991} for details) acts on it identically this equation implies $i_{[s,s]}\st{n}{\omega} = 0$.
\end{proof}

It turns out that a presymplectic local BV system still defines a local BV system on a jet-bundle of a symplectic quotient of $\cE$, at least locally. More precisely:
\begin{prop}\label{prop:p-BV-quotient}
Given a presymplectic local BV system,  assume in addition that the kernel distribution is regular. Then, at least locally,  this data naturally induces a local BV system with the underlying bundle being the symplectic quotient of $\cE$ by the kernel distribution $\cK$. 
\end{prop}
\begin{proof}
The proof is based on the explicit construction. At first step we employ Lemma~\bref{prop:bundle_quotient}  to construct the respective quotient which is itself a new fiber bundle $\cG \to X$, which is well defined at least locally.
This comes together with the projection $\pi_\cG:\cE \to \cG$ which is fiberwise. 
Restricting to the local analysis, we can then realise $\cG$ as a subbundle in $\cE$ which in turn \rchanged{induces}
the embedding of $J^\infty(\cG)\subset J^\infty(\cE)$ as a subbundle. In particular,
the pullback by the embedding commutes with the horizontal differential $\dh$.
Note, however, that the embedding is not canonical. 

The projection $\pi_\cG:\cE \to \cG$ induces a projection  
$J^\infty(\cE) \to J^\infty(\cG)$ of the respective jet bundles. Together with the embedding it defines a projection $\tilde{s}$ of $s$ to $J^\infty(\cG)$. The projection is generally not evolutionary (because $s$ is not tangent to $J^\infty(\cG)$, in general). By pulling back \eqref{p-master} to $J^\infty(\cG)$ one finds:
\begin{equation} 
\label{s-inv-red}
i_{\tilde{s}} i_{\tilde{s}} (i^* \st{n}{\omega})=\dh(\cdot)\,,
\end{equation}
where $i$ denotes the embedding $J^\infty(\cG) \to J^\infty(\cE)$. Indeed,
$i_{\tilde s} (i^*\st{n}{\omega})=i^*(i_{s} \st{n}{\omega})$ because projection can only add a vector from the kernel. Then $i^*(\dh f)=\dh (i^* f)$ for any $f\in \bigwedge^\bullet(J^\infty(\cE))$. The same arguments show that 
the pull-back of \eqref{eq:s-inv-2} gives:
\begin{equation}
\label{p-master-red}
    i_{\tilde s} i^*\st{n}{\omega}+\dv i^* \st{n}{\cL}=\dh(\cdot) \,.
\end{equation}

Observe that $i_{\tilde s} i_{\tilde s} (i^* \st{n}{\omega})$ depends only on the action of $\tilde s$ on zero jets because $\st{n}{\omega}$ is itself a pullback to $J^\infty(\cG)$ of the form on $\cG$. It follows, the equation is also fulfilled if $\tilde s$ is replaced by the evolutionary vector field $s^\prime$ which acts on zeroth jets just like $\tilde s$, i.e. if $\psi^\alpha$ are coordinates on the fibers of $\cG$
pulled back to $J^\infty(\cG)$ then $s^\prime$ is uniquely determined by $s^\prime \psi^\alpha=\tilde s\psi^\alpha$ and that $s^\prime$ is evolutionary. Finally, \rchanged{Lemma}~\bref{prop:nilp-mod-kernel}
implies that an evolutionary vector field $s^\prime$ of degree $1$ satisfying~\eqref{s-inv-red} and \eqref{p-master-red} is nilpotent provided $i^*\st{n}{\omega}$ is nondegenerate and hence  this data defines a local BV system.
\end{proof}

\subsection{Weak presymplectic gPDEs}

\begin{definition} \label{weak_presymp_g_PDE}
A weak presymplectic gauge PDE is a $Z$ graded fiber bundle $E \xrightarrow[]{\pi} T[1]X$, \rchanged{$\dim(X)=\dmsn$},
equipped with a 2-form $\omega$ of degree $n-1$, a 0-form $\cL$ of degree $n$ and a vector field $Q$ of degree 1 satisfying $Q\circ \pi^* = \pi^* \circ \dx$ and 
\begin{equation}
\label{wgPDE}
    d\omega = 0, \qquad i_Q\omega + d\mathcal{L} \in \mathcal{I}, \qquad \frac{1}{2} i_Q i_Q\omega + Q\mathcal{L} = 0\,.
\end{equation}
\rchanged{where $\cI \subset \bigwedge^\bullet(E)$ is the ideal generated by elements of the form $\pi^*\alpha, \alpha \in \bigwedge^{k > 0}(T[1]X)$.}
\end{definition}
This version of the definition is a slight variation of that given in~\cite{Grigoriev:2022zlq}. Note that similar axioms originally appeared in~\cite{Alkalaev:2013hta}. This data still defines an action functional~\eqref{pgPDE-action} and, in fact, a presymplectic local BV system. \rchanged{It is important to mention that at this level we do not ask $\omega$ to be regular.}

Solutions to this system can be defined as solutions of the Euler-Lagrange equations determined by~\eqref{pgPDE-action} and understood modulo the natural algebraic gauge equivalence generated by the vector fields from the kernel of $\omega$. Of course one can equally well say that the physical interpretation of the system is that of the induced local BV system.

As an alternative one can define such objects by specifying, instead of $\omega$, a presymplectric potential $\chi$ so that $\omega=d\chi$. This fixes the boundary terms in the action functional~\eqref{pgPDE-action}. Moreover, it is easy to check that in this case adding to $\chi$ a 1-form from $\cI$ accompanied by a suitable transformation of $\cL$ is a symmetry of the equations~\eqref{wgPDE} and also does not affect the action~\eqref{pgPDE-action} modulo a field-independent boundary term. To see this note that if $\omega$ and $Q$ are given, the ambiguity in $\cL$
satisfying the second condition from~\eqref{wgPDE} is given by functions of the form $\pi^*(h), h\in\cC^\infty(T[1]X)$ so that the integrand of~\eqref{pgPDE-action} can only change by $\pi^*(h)$.

An important fact which justifies the above definition is that this geometrical object naturally defines a local BV system provided certain regularity condition is imposed on $\omega$. The crucial step in the construction is the passage to the super-jet bundle $SJ^\infty(E)$. This can be related  with the jet-bundle of another fiber bundle $\bar E$ whose typical fiber $\bar{F}$ is the space of super-maps from $T_x[1] X$ to the typical fiber $F$ of $E$. More precisely, $J^\infty(\bar E)$ can be understood as a pullback of $SJ^\infty(E) \to T[1]X$ to the zero section of $T[1]X$. Note that $\bar{F}$ is finite-dimensional provided $F$ is. General discussion of super jet-bundles can be found~\cite{Khudaverdian:2001qe}, see also~\cite{Grigoriev:2019ojp,Grigoriev:2022zlq} for precisely the present setup.

\rchanged{The prolongation of $Q$ to $SJ^\infty(E)$ restricts to $J^\infty(\bar E)\subset SJ^\infty(E)$ and is a vertical evolutionary vector field $s$. The presymplectic structure $\omega$ on $E$ induces a 2-form
on each fiber of $\bar E$. Indeed, for a given $x\in X$ let us
consider the fiber $E_x$ of $E$ over $x$ (here we treat $E$ as a bundle over $X$). $E_x$ itself is a bundle over $T_x[1]X$. The fiber $\bar E_x$ of $\bar E$ at $x$ is by definition the space of supersections of $E_x$ over $T_x[1]X$. Furthermore, restriction of $\omega$ to $E_x$ gives rise to a 2 form on $\bar E_x$ via the standard construction which is~\eqref{indomega} restricted to $T_x[1]E$. If this 2-form is regular and its rank does not depend on $x$ we say that $\omega$ is quasi-regular. Note that nonregular $\omega$ may be quasi-regular and this is precisely what happens in applications. As we are going to see if $\omega$ is quasi-regular, one can take a symplectic quotient (at least locally), resulting in a local BV system.} The construction is a more refined version of that originally presented in~\cite{Grigoriev:2022zlq}. 

Having explained the idea of the construction let us formulate it as a theorem and give a proof:
\rchanged{\begin{theorem}
    Let $(E,T[1]X, Q, \omega)$ be a weak presymplectic gPDE. If $\omega$ is quasi-regular this data determines a local BV system, at least locally.
\end{theorem}}
\begin{proof}
Given a weak presymplectic gPDE $(E,T[1]X, Q, \omega)$, the system \newline $(SJ^\infty(E), T[1]X, \Bar{Q}, \Bar{\omega}, \Bar{\cL})$ is also a weak presymplectic gauge PDE. Here $\Bar{Q}$ is a unique 
prolongation of $Q$ and $\bar\omega =(\pi^\infty_E)^*\omega, \bar\cL= (\pi^\infty_E)^*\cL$, where $\pi^\infty_E$
is a canonical projection $SJ^\infty(E)\to E$. It follows, that conditions \eqref{wgPDE} imply:
\begin{equation} \label{weak_cond}
    d\Bar{\omega} = 0, \qquad i_{\Bar{Q}}\Bar{\omega} + d\Bar{\cL} \in \cI, \qquad \frac{1}{2}i_{\Bar{Q}}i_{\Bar{Q}}\Bar{\omega} + \Bar{Q}\Bar{\mathcal{L}} = 0 \,,
\end{equation}
where \rchanged{by abuse} of notations
we denote by $\cI$ the ideal in $\bigwedge^\bullet(SJ^\infty(E))$, generated by forms of positive degree on the base $T[1]X$.

First two conditions from \eqref{weak_cond} imply $\Bar{\omega} = d\Bar{\chi}$ and 
\begin{equation} \label{prolonged_main}
    i_{\Bar{Q}}\Bar{\omega} + d\Bar{\mathcal{L}} + \mathcal{I} = 0\,.
\end{equation}
Since $Q$ projects to $\dx$ the prolonged vector field $\Bar{Q}$ can be decomposed as $\Bar{Q} = D + s$ with $D = \textrm{tot}(d_X)$ and $s$ evolutionary.  
We have $i_D d\Bar{\chi} = L_D\Bar{\chi} + d i_D\Bar{\chi}$ 
and hence \eqref{prolonged_main} can be rewritten as
\begin{equation}
    i_s \Bar{\omega} = -d(\bar\cL + i_D\Bar{\chi}) - L_D(\Bar{\chi}) + \mathcal{I} \,.
\end{equation}
Now we recall that on a jet-bundle the projection to the vertical part is well-defined. Taking the vertical part of the above equation we get:
\begin{equation} \label{vert}
    i_s \Bar{\omega}^v = -\dv(\bar\cL + i_D\Bar{\chi}) - L_D(\Bar{\chi^v})\,.
\end{equation}

As it is clear from the structure of~\eqref{vert}, it has a chance to encode a condition that $s$ is the Hamiltonian vector field, which is one of the axioms of a presymplectic local BV system. Another axiom is the pre-master equation. To obtain it consider the following expression:
\begin{equation} \label{i_si_s}
    i_si_s\Bar{\omega} = i_{\Bar{Q}-D}i_{\Bar{Q}-D}\Bar{\omega} = i_{\Bar{Q}}i_{\Bar{Q}} \Bar{\omega} - 2i_{\Bar{Q}}i_D\Bar{\omega} + i_Di_D\Bar{\omega} \,.
\end{equation}
Let us rewrite \eqref{prolonged_main} as $i_{\Bar{Q}}\Bar{\omega} + d\Bar{\mathcal{L}} + \alpha = 0, \alpha \in \mathcal{I}$. Applying $i_{\Bar{Q}}$ on it we get $i_{\Bar{Q}}i_{\Bar{Q}}\Bar{\omega} + Q\Bar{\mathcal{L}} + i_{\Bar{Q}}\alpha = 0$. Since $\alpha \in \mathcal{I}$ is a 1-form $i_{\Bar{Q}}\alpha = i_D\alpha$. Making use of the third equation in \eqref{weak_cond} we get $i_D\alpha = -\frac{1}{2}i_{\Bar{Q}}i_{\Bar{Q}} \Bar{\omega}$ Then the second term on the r.h.s. of \eqref{i_si_s} is $2i_D(d\Bar{\mathcal{L}} + \alpha)= 2D(\Bar{\mathcal{L}}) + 2i_D\alpha = 2D(\Bar{\mathcal{L}}) - i_{\Bar{Q}}i_{\Bar{Q}} \Bar{\omega}$ which cancels the first term in \eqref{i_si_s}. The third term is $i_Di_Dd\Bar{\chi} = i_DL_D\Bar{\chi} + i_D d i_D \Bar{\chi} = 2D(i_D\Bar{\chi})$. 
So all in all we have:
\begin{equation} \label{p-master-eq}
    \half i_s i_s\Bar{\omega} = D(i_D\Bar{\chi} + \mathcal{L})\,.
\end{equation}

Let us introduce a special coordinate system on $SJ^\infty(E)$, where the total derivative with respect to $\theta$ is simply $D^\theta_a=\dl{\theta^a}$, see Appendix~\bref{app:jets} for more details. 
In this coordinate system evolutionary vector fields commute with $\dl{\theta^a}$ and hence we can expand our equations in powers of  $\theta^a$, i.e. in horizontal form degree if we switch to the language of differential forms on $J^\infty(\bar{E})$. Introduce notations $\st{k}{\gtj\omega}$, $\st{k}{\gtj\chi}$, $\st{k}{\gtj\cL}$ for the horizontal form degree $k$ components of $\Upsilon({\bar\omega}{}^v)$, $\Upsilon({\bar\chi}{}^v)$, $\Upsilon({\bar \cL}{}^v)$, respectively. 
In terms of the differential forms on $J^\infty(\bar{E})$ the horizontal form-degree $n$ components of equations \eqref{vert} and \eqref{p-master-eq} read as
\begin{equation}
    i_s \st{n}{\gtj{\omega}}+\dv \st{n}{\gtj{\cL}}_{BV} = -\dh \st{n-1}{\gtj{\chi}}\,.
\end{equation}
where $\gtj{\cL}_{BV}=\Upsilon(\cL) + \Upsilon(i_D\Bar{\chi})$ and
\begin{equation}
 \half i_s i_s\st{n}{\gtj\omega} = \dh(\st{n-1}{\tilde\cL_{BV}})\,.
\end{equation}
These are precisely the defining relation of a presymplectic local BV system on $J^\infty(\bar{E})$.

Finally, $\st{n}{\gtj\omega}$ is a pullback to $J^\infty(\bar{E})$ of an $(n+2)$-form on $\bar{E}$ and because $\st{n}{\gtj\omega}$ is proportional to the volume form $(dx)^n$, its kernel distribution is vertical. \rchanged{The quasi-regularity of $\omega$ means that $\st{n}{\gtj\omega}$ is regular, therefore all the conditions from Definition~\bref{def:p-BV} are satisfied and hence Proposition~\bref{prop:p-BV-quotient} implies that, at least locally, a weak presymplectic gPDE determines a local BV system.}
\end{proof}

\subsection{Weak presymplectic PDEs and multisymplectic systems}

Before switching to the examples of weak presymplectic gPDEs associated to various gauge field theories let us consider the special case, where  $Q$ does not encode any gauge transformations. More specifically, consider a presymplectic gPDE $(E,Q,T[1]X,\omega)$ such that  the only nonvanishing degree variables are $\theta^a$, $\gh{\theta^a}=1$ or, in other words, $\fZ$-degree is horizontal. This means that the underlying gPDE is actually a usual PDE. Indeed,  in this case $\cC^{\infty}(E)$ can be identified with the algebra of horizontal differential forms on the underlying bundle $\gtj{E}$ while $Q$ corresponds to the horizontal differential $\dh$ on $\gtj{E}$, cf.~\eqref{isom2}. Indeed, by the degree reasoning $Q$ is linear in $\theta^a$ and can be seen as a flat covariant differential on $\gtj{E} \to X$ so that $(\gtj{E},\dh)$ is a PDE defined in the geometrical way, see e.g.~\cite{Krasil'shchik:2010ij}. More detailed discussion of the above can be found in~\cite{Grigoriev:2019ojp}.

In this case there are no nonvanishing gauge parameters (=vertical vector fields on $E$ of degree $-1$) and hence the underlying $(E,Q,T[1]X,\omega)$ does not encode nontrivial gauge (as well as gauge for gauge) transformations.  Note that it does not mean that the underlying PDE does not admit gauge symmetries, it only means that there are no gauge symmetries accounted by $Q$. In this setup $gh(\symp)=\dmsn-1$ implies  that in local coordinates $x^a,\theta^a,\psi^A$ the expression for $\symp$ has the following structure  
\begin{equation}
\symp=\symp_{AB}^ad\psi^Ad\psi^B(\theta)^{(\dmsn-1)}_a\,.
\end{equation}

Let us now reformulate the system in terms of the differential forms on $\gtj{E}$ which is a pullback of $E$ to the zero section of $T[1]X$. Alternatively, $\gtj{E}$ can be also understood as a submanifold in $E$ singled out by $\theta^a=0$. Using the map $\Upsilon$ defined in~\eqref{isom} from forms on $E$ to forms on $\gtj{E}$ we define
\begin{equation}
\begin{gathered}
\gtj{\omega}=\Upsilon(\omega)\subset \bigwedge\nolimits^{n-1,2}(\gtj{E})\,,\qquad  
\gtj{\chi}=\Upsilon(\chi)\subset \bigwedge\nolimits^{n-1,1}(\gtj{E})\,,\\
\gtj{\cL}=\Upsilon(\cL)\subset \bigwedge\nolimits^{n,0}(\gtj{E})\,.
\end{gathered}
\end{equation}
Using the properties~\eqref{isom2} of $\Upsilon$ and $L_Q\omega \subset \cI$ one then finds:
\begin{equation}
\label{intrinsic-gPDE}
\dv\gtj{\symp}=0\,,\qquad\dh\gtj{\symp}=0\,, \qquad \gtj{\symp}=\dv \gtj{\chi}\,.
\end{equation}

In this way we have observed that $\gtj{\symp}$ can be indeed identified with a compatible presymplectic current~\cite{Kijowski:1973gi,Kijowski:1979dj,Crnkovic:1986ex,Lee:1990nz}\footnote{We use the term presymplectic current to indicate that this is an $n-1$ horizontal and $2$-vertical form on the equation manifold.} on $\gtj{E}$, i.e. we are indeed dealing with a PDE equipped with a compatible presymplectic current. Because cohomology of $\dv$ can only originates from the global geometry of the fibres, at least locally there  exists \rchanged{$(n,0)$-form
$\gtj{l}$}
such that $\gtj{\omega}=d(\gtj{\chi}+\gtj{l})$ and there is a natural action functional on the space of sections of $\gtj E$:
\begin{equation}
\label{intrinsic}
S^C[\sect^*]=\int\sect^*(\gtj{\chi}+\gtj{l})\,.
\end{equation}
\rchanged{Of course, the functional is also defined locally, in general.} In the PDE geometry setup this action was proposed in~\cite{Grigoriev:2016wmk,Grigoriev:2021wgw} as a way to encode the Lagrangian formulation in terms of the geometry of the equation manifold, see also~\cite{Khavkine2012,Druzhkov:2021}. Note that the action functionals analogous to \eqref{intrinsic} are well known in the context of multisymplectic systems, see e.g.~\cite{Gotay:1997eg,Bridges2009,Rom_n_Roy_2009} and the discussion below.

Let us now show that~\eqref{intrinsic} is in fact the same as the AKSZ-like action~\eqref{pgPDE-action}. To see this
recall $i_Q\omega+d \cL\in\cI$ can be rewritten as  $L_Q\chi+d(\cL+i_Q\chi)\in \cI$, giving
\begin{equation}
\label{covHam1}
    \dh\gtj{\chi} +\dv(\Upsilon(i_Q\chi)+\gtj{\cL})=0\,,
\end{equation}
upon applying $\Upsilon$ to the both sides. It follows, $\gtj{\omega}=d(\gtj{\chi}+\gtj{l}^\prime)$, where $\gtj{l}^\prime=\Upsilon(i_Q\chi)+\gtj{\cL}$, so that one can assume $\gtj{l}^\prime=\gtj{l}$  and we have recovered all the ingredients of the action~\eqref{intrinsic} from the presymplectic gPDE data.

Finally, introducing the coordinates $x^a,\theta^a,\psi^A$ on $E$, action~\eqref{pgPDE-action} can be written as:
\begin{multline}
S^{gPDE}[\sect^*]=\int_{T[1]X}(\sect^*({\chi}_A)\dx \sigma^*(\psi^A)+\sect^*(l-{\chi}_A Q\psi^A))=\\
\int_{X}(\sect^*(\gtj{\chi}_A) d \sigma^*(\psi^A)+\sigma^*(\gtj{l})-\sect^*(\gtj{\chi}_A\dh \psi^A))\,,
\end{multline}
where $l \in \cC^{\infty}(E)$ is a unique function such that  $\Upsilon (l)=\gtj{l}$ and by some abuse of notations we use $\sigma$ to denote a section
$X \to E$ and its unique extension to a section $T[1]X \to E$.  At the same time, in terms of the coordinates $x^a,\theta^a,\psi^A$,  the action~\eqref{intrinsic} takes the following form 
\begin{multline}
S^C[\sect^*]=\int_X \sect^*(\gtj{\chi}_A\dv\psi^A+ \gtj{l})=\int_X\sect^*(\gtj{\chi}_A(d-\dh)\psi^A+\gtj{l})=\\
=\int_X \sect^*(\gtj{\chi}_Ad\psi^A)+\sect^*(\gtj{l}- \gtj{\chi}_A\dh\psi^A))\,,
\end{multline}
so that indeed the actions coincide.

Let us also consider a weak presymplectic gauge PDE $(E,Q,T[1]X,\omega)$ such that $\gh{\psi^A}=0$ and see what it corresponds to in terms of the geometry of the underlying fiber bundle $\gtj{E}$.  In this case $Q$ is still linear in $\theta^a$ but $Q^2$ is generally nonvanishing. From the point of view of $\gtj{E}$ it is a covariant differential, which in contrast to the case considered above, is generally non-flat.  Note that the third condition \eqref{wgPDE} is trivially satisfied because $\gh{i_Qi_Q\omega}=\gh{Q\cL}=n+1$ and hence does not impose any constraints on $Q^2$. Nevertheless, all the above reformulations of the action can be repeated. Namely, the action~\eqref{pgPDE-action} can be rewritten as
\begin{equation}
S[\sigma]=\int_X \sigma^*(\Theta)\,, \qquad \Theta=\Upsilon(\chi+\cL+i_Q\chi)\,.
\end{equation}
Note that $\Theta$ is an $n$-form on $\gtj{E}$.

From the above discussion it is clear that we have arrived at a multisymplectic system, provided $d\Theta$ is nondegenerate. Recall, that a multisymplectic system is a bundle $\gtj{E}\to X$ equipped with an $n$-form $\Theta$, $n=\dim{X}$, such that $\Omega=d\Theta$ vanishes on any 3 vertical vectors, see e.g.~\rchanged{\cite{Gotay:1997eg, BRIDGES_1997,Hydon1627,deLeon:2005hu,Bridges2009,Gaset:2022ema}}\footnote{\rchanged{Note that there are various versions of the definition. For instance the multisymplectic bundle as defined in~\cite{Gotay:1997eg} contains the above multisymplectic system as a subbundle defined by the Legendre transform.}}. The last condition is automatically satisfied in our case because $\Theta$ originates from a sum of 1 and 0-forms on $E$. Further details on multisymplectic systems can be found in e.g.~\cite{Kijowski:1979dj,Aldaya:1980zz,Gotay:1997eg,Roman-Roy:2005vwe,Gaset:2022ema}.

It is important to stress, however, that the multisymplectic system determined by a weak presymplectic gPDE comes equipped with the additional structure, the Ehresmann connection on $\gtj{E}$ encoded in $Q$, which is not necessarily flat. This system can be seen as the quotient of the respective PDE by the maximal regular subdistribution of the vertical kernel distribution of $\gtj{\omega}$. The differential induced by $\dh$ on the quotient generally fails to be nilpotent and hence its associated connection is not flat, in general.


\section{Examples}
\label{sec:examples}
\subsection{$p$-forms}

The gauge field of the $p$-form theory is a spacetime $p$-form $A$ subject to the gauge transformation $\delta A=dB$ and the Lagrangian of the form $(dA)^2$.
For $p>1$ this is a reducible gauge theory and its proper BV formulation involves $p$ generation of ghosts (for ghosts) together with their conjugate antifields.
It turns out that the complete BV formulation of this system arises from a rather concise weak presymplectic gPDE. In this section we use notation $T^{(p)}$ to denote a totally antisymmetric tensor $T^{a_1...a_n}$ and $T^{(k)}\cdot R_{(k)}$ to denote a contraction of such tensors.

Consider a trivial bundle $E \to T[1]X$ with $X$ being the $n$-dimenisonal Minskowski space and the fiber being the space with coordinates $C$, $\gh{C}=p$ and  $F^{(p+1)}$, $\gh{F^{(p+1)}}=0$.  It is convenient to introduce ``generating function'' $\bF=\frac{1}{(p+1)!}F_{a_1\ldots a_{p+1}} \theta^{a_1}\ldots \theta^{a_{p+1}}$ in terms of which the $Q$-structure is defined as:
\begin{equation}
QC=\bF \,, \qquad 
QF^{(p+1)}=0\rchanged{=Q\bF}\,.
\end{equation}
The presymplectic structure is given by 
\begin{equation}
\symp=d\chi\,, \qquad \chi=\frac{1}{2}(\star \bF)dC\,,
\end{equation}
where $\star$ denotes the usual Hodge conjugation of functions in $\theta^a$ seen as exterior forms, i.e. $(\star \bF)=\frac{1}{(n-p-1)!}\epsilon_{(n)}\cdot (F^{(p+1)}\theta^{(n-p-1)})$. 

It is easy to check that the axioms are indeed fulfilled and $\cL$ defined through
$i_Q\symp+d\cL\in \mathcal{I}$ can be taken as
\begin{equation}
\cL=-\frac{1}{4}\bF (\star \bF)\,.
\end{equation}
Note that in this case $Q^2=0=Q\cL$.

Introducing fields $A_{(p)}(x)$ and ${\mathrm F}_{(p+1)}(x)$ to parameterize sections so that $\sigma^*(C)=A_{(p)}\cdot \theta^{(p)}$ and $\sigma^*(F_{(p+1)})=\mathrm{F}_{(p+1)}$ the action functional \eqref{pgPDE-action} takes the form
\begin{multline}
\label{actionPN}
S[A_{(p)},\mathrm{F}^{(p+1)}]
=\int_{T[1]X} (\dx \bA)(\star \bF)-\frac{1}{4}\bF(\star \bF)\\
=\int d^\dmsn x(\frac{1}{2}\d_{[(1)}A_{(p)]}\cdot \mathrm{F}^{(p+1)}-\frac{1}{4}\mathrm{F}^{(p+1)}\cdot \mathrm{F}_{(p+1)})\,.
\end{multline}
where we also introduced the following generating function $\bA=\frac{1}{p!}A_{(p)}\theta^{(p)}$. This action is equivalent to the standard $(dA)^2$ action through the elimination of the auxiliary field $\mathrm{F}_{(p+1)}$.

It is instructive to see how the formalism also generates the BV action for this model. For simplicity let us restrict to the simplest case $p=2$, $\dmsn=4$. The component expression for the symplectic structure is given by
\begin{equation}
\symp=\frac{1}{2}\epsilon_{abcd}(dF^{abc}\theta^d)dC\,.
\end{equation}
Coordinates on the fibers of $\gtj{E}$ are introduced as components of supersection $\hat\sigma:T_x[1]X\to E$:
\begin{equation}
\begin{aligned}
\Hat{\sect}^*(C)&=\st{0}{C}+\st{1}{C}_a\theta^a+\frac{1}{2}A_{ab}\theta^a\theta^b+\frac{1}{6}\st{3}{C}_{abc}\theta^a\theta^b\theta^c+\ldots\,, \\
\Hat{\sect}^*(F^{abc})&=F^{abc}+\st{1}{F}{}^{abc}{}_d\theta^d+\frac{1}{2}\st{2}{F}{}^{abc}{}_{de}\theta^d\theta^e+\frac{1}{6}\st{3}{F}{}^{abc}{}_{def}\theta^d\theta^e\theta^f+\ldots\,.
\end{aligned}
\end{equation}

Presymplectic structure on $\gtj{E}$ has the following form:
\begin{equation}
\label{symp2form}
\st{4}{\gtj{\symp}}=\frac{1}{2}(dCdC^*-dC_adC^{*}{}^{a}+dA_{ab}dA^*{}^{ab}-dF^*_{abc}dF^{abc})(dx)^4\,,
\end{equation}
where $C=\st{0}{C}$, $C^*=\st{3}{F}{}^{abc}{}_{abc}$, $C_a=\st{1}{C}_a$, $C^{*}{}^{a}=3\st{2}{F}{}^{abc}{}_{bc}$, $A^*{}^{ab}=3\st{1}{F}{}^{abc}{}_c$, $F^*_{abc}=\st{3}{C}_{abc}$ and we have switched to the language of differential forms on $\gtj{E}$. 

It follows that we have indeed recovered the fiber bundle $\gtj{E}$ underlying the standard BV formulation of the $p$-form theory, see e.g.~\cite{HT-book}. Indeed,
the symplectic structure is canonical and the spectrum of ghost degrees is summarized in the following table:
\begin{center}
\begin{tabular}{ |c|c|c|c|c|c|c|c|c|c| } 
 \hline
 field & $A_{ab}$ & $A^*{}^{ab}$ &$F^{abc}$&$F^*_{abc}$&$C_a$&$C^{*}{}^{a}$&$C$&$C^*$ \\ 
 \hline
 $gh(\cdot)$ & $0$ & $-1$ &$0$&$-1$&$1$&$-2$&$2$&$-3$ \\ 
 \hline
\end{tabular}
\end{center}
All the remaining variables are in the kernel of $\st{4}{\gtj{\symp}}$ and are factored out. Finally, in the above coordinates the BV-AKSZ action~\eqref{AKSZ-action} is given by  
\begin{multline}
S_{BV}=\int d^4 x(\frac{1}{2}F^{abc}(\d_aA_{bc}+\d_bA_{ca}+\d_cA_{ab})-\frac{1}{4}F^{abc}F_{abc}+\\
+\frac{1}{2}A^*{}^{ab}(\d_aC_b-\d_bC_a)+\frac{1}{2}C^*_a\d^aC)
\end{multline}
and is indeed a standard BV action of the $p$-form theory.

\subsection{Freedman-Townsend model}
In this section we present a weak presymplectic gPDE formulation of the Freedman-Townsend model~\cite{Freedman:1980us}. Its first-order action has the following form:
\begin{equation}
\label{actionFT}
S[B^{ab},A_a]=\int d^4 xTr(-\frac{1}{4}\epsilon_{abcd}B^{ab}(\d^cA^d-\d^dA^c+[A^c,A^d])+\frac{1}{4}A^aA_a)\,.
\end{equation}
Its gauge transformations can be written as
\begin{equation}
\delta B^{ab}=\d^a\lambda^b-\d^b\lambda^a+[A^a,\lambda^b]-[A^b,\lambda^a]\,, \qquad 
\delta A^a = 0\,.
\end{equation}

To see that \eqref{actionFT} is indeed a consistent deformation of the $2$-form theory discussed in the previous Section, one can use the following field redefinition: $A_a=\epsilon_{abcd}F^{bcd}$. Keeping only quadratic terms one indeed recovers \eqref{actionPN} with $p=2$, $\dmsn=4$.

The weak presymplectic gPDE formulation of this model is constructed as follows. The fiber of the underlying bundle $E \to T[1]X$ is the space with matrix-valued coordinates $C$, $gh(C)=2$ and $A_a$, $gh(A_a)=0$. Presymplectic structure and $Q$-structure are given by:
\begin{equation}
\symp = d\chi\,, \qquad \chi= -\frac{1}{2}Tr(C dA_a\theta^a)\,.
\end{equation}
\begin{align}
QC &= -\frac{1}{6}\epsilon_{abcd}A^a\theta^b\theta^c\theta^d - 2[C,A_a]\theta^a \,, \\
QA_a &= -[A_a,A_b]\theta^b\,.
\end{align}

Hamiltonian $\cL$ defined defined by $i_Q \omega+d\cL\in \cI$ can be taken as 
\begin{equation}
\cL=\frac{1}{2}Tr(\frac{1}{2}A^aA_a(\theta)^{4}+C[A_a,A_b]\theta^a\theta^b) \,.
\end{equation}
It is easy to check that the remaining conditions are satisfied
$i_Qi_Q\symp=Q\mathcal{L}=0$.

Introducing coordinates on the fibers of $\gtj{E}$ via:
\begin{equation}
\begin{aligned}
\Hat{\sect}^*(C)&=\st{0}{C}+\st{1}{C}_a\theta^a+\frac{1}{2}B_{ab}\theta^a\theta^b+\frac{1}{6}\st{3}{C}_{abc}\theta^a\theta^b\theta^c+\ldots\,,\\
\Hat{\sect}^*(A^a)&=A^a+\st{1}{A}{}^a{}_b\theta^b+\frac{1}{2}\st{2}{A}{}^a{}_{bc}\theta^b\theta^c+\frac{1}{6}\st{3}{A}{}^a{}_{bcd}\theta^b\theta^c\theta^d+\ldots\,,
\end{aligned}
\end{equation}
the symplectic structure $\st{4}{\gtj{\omega}}$ on $\gtj{E}$ takes the form:
\begin{equation}
\st{4}{\gtj{\symp}}=\frac{1}{2}\mathrm{Tr}\,(\frac{1}{6}dCdC^*+\frac{1}{2}dC^adC^*_a+\frac{1}{6}dA^*_adA^a+\frac{1}{4}dB^{ab}dB^*_{ab})(dx)^4 \,,
\end{equation}
where now $C=\st{0}{C}$, $C^*=\epsilon^{abcd}\st{3}{A}_{abcd}$, $C^a=\st{1}{C}{}^a$, $C^*_a=\epsilon_{abcd}\st{2}{A}{}^{bcd}$, $A^*_a=\epsilon_{abcd}\st{3}{C}{}^{bcd}$ and $B^*_{ab}=\st{1}{A}{}_{ab}$. It is of course not surprising that this is just a matrix version of~\eqref{symp2form} because the interaction do not affect the symplectic structure.

Finally, the BV-AKSZ action~\eqref{AKSZ-action} is given by  
\begin{multline}
S_{BV}=\int d^4 x \mathrm{Tr}(-\frac{1}{4}\epsilon_{abcd}B^{ab}(\d^cA^d-\d^dA^c+[A^c,A^d])+\frac{1}{4}A^aA_a -\\
-\frac{1}{4}\epsilon^{abcd}C_a\d_{[b}B^*_{cd]}-\frac{1}{4}C\d^aC^*_a)\,,
\end{multline}
so that we have indeed recovered the standard BV formulation of the Freedman-Townsend theory~\cite{Batlle:1988rd}.

\subsection{Chiral Yang-Mills theory}

The so called chiral and (anti-)selfdual theories are extensively studied in the literature \cite{Chalmers:1996rq,Witten:2003nn,Abou-Zeid:2005zfo,Krasnov:2016emc,Sen:2019qit,Ponomarev:2022qkx,Herfray:2022prf,Basile:2022mif,Cattaneo:2023cnt,Hull:2023dgp,Sharapov:2021drr}. 
A simple  example of a chiral (selfdual) theory can be constructed starting from the  Yang-Mills theory. Let us first give an overview of YM formulations in terms of the (anti-)selfdual curvatures:
\begin{equation}
    (\star F^{\mp})_{ab}=\mp F^{\mp}_{ab} \,, \qquad (\star F)_{ab} = \half \epsilon_{abcd}F^{cd}\,.
\end{equation}
In Euclidean signature $\star^2 = 1$ and therefore has eigenvalues $\pm 1$.
To avoid complexification  we will stick to Euclidean signature for this section. Note that $F^\pm_{ab} = \half( F_{ab} \pm \star F_{ab})$.

The action of YM theory can be rewritten in terms of $F^{\pm}$ as:
\begin{multline} \label{YM-act}
    S_{YM}[A] = \int \tr ( F \wedge \star F) =\\=
    \int d^4x \tr ( F_{ab}F_{pk}) \epsilon^{abcd}\epsilon^{pk}{ }_{cd} = \int d^4x \tr ( (F^+)^2 + (F^-)^2)
\end{multline}
In its turn, the action of the topological YM theory can be written as:
\begin{equation} \label{tYM-act}
    S_{tYM}[A] = \int \tr(F\wedge F) = \int d^4x \tr( (F^+)^2 - (F^-)^2)\,.
\end{equation}
Taking the sum of the actions~\eqref{YM-act} and \eqref{tYM-act} one arrives at the equivalent (modulo boundary terms) action $\int (F^+)^2$ which involves only selfdual part of the curvature.  It can be further rewritten using an additional field $G_{ab}$ which is selfdual,
\begin{equation} \label{F^+^2}
    S_{cYM}[G,A] = \int d^4x \tr(F_{ab}G^{ab} + \frac{1}{2}G^{ab}G_{ab}).
\end{equation}
Integrating out $G$ gives $(F^+)^2$. Finally, the Chalmers--Siegel action~\eqref{C-S-action}, is obtained by omitting the $G^{ab}G_{ab}$ term in~\eqref{F^+^2}~\cite{Chalmers:1996rq}.

Both the chiral form of YM theory and the Chalmers-Siegel admit concise formulations in terms of weak presymplectic gPDEs. Let us start from
the graded fiber bundle $E\to T[1]X$ underlying the respective formulation of the standard YM theory~\cite{Grigoriev:2022zlq}:
\begin{equation}
\begin{gathered}
    x^a,\theta^a,\quad C\,,\quad \gh{C}=1\\
    G_{ab},\quad \gh{G_{ab}}=0\,.
\end{gathered}
\end{equation}
Here $G_{ab} = -G_{ba}$ and $G_{ab}$ and $C$ take values in some Lie algebra  $\mathfrak{g}$ which we assume semisimple. Let $Q$ act as:
\begin{equation}
\begin{gathered}
    Qx^a=\theta^a\,\\
    QC=-\half[C,C]  + \alpha (G_{ab})^+ \theta^a\theta^b  + \beta G_{ab}^- \theta^a\theta^b \,,
    \\
    QG_{ab}=[G_{ab},C]\,,
\end{gathered}
\end{equation}
where $\alpha, \beta \in \fR$ are some coefficients and $\commut{\cdot}{\cdot}$ denotes the graded commutator in $\algg$ tensored with $\cC^\infty(E)$. As a  presymplectic structure we take:
\begin{equation} 
    \omega = \tr(d(G_{ab}^+\theta^a\theta^b)dC + \gamma d(G_{ab}^-\theta^a \theta^b)dC)
\end{equation}
which is compatible with the differential:
\begin{multline}
    i_Q \omega = -d\cL + \cI = \tr(d(-G^+_{ab} \half[C,C]\theta^a \theta^b + \frac{\alpha}{2}G^+_{ab} (G^{ab})^+ \theta^4  - \\ -
    \gamma  G^-_{ab}\theta^a \theta^b \half [C,C]  + \frac{\beta \gamma}{2}G^-_{ab} (G^{ab})^-\theta^4)) + \cI
\end{multline}
Moreover, it is easy to check that $i_Qi_Q \omega = 0 = \rchanged{Q\cL}$ so that the axioms are fulfilled. 

Let us parameterize the space of sections by
\begin{equation}
    \begin{gathered}
        \sigma^*(C) = A_a(x) \theta^a \,,\\
        \sigma^*(G_{ab}) = G_{ab}(x) \,.
    \end{gathered}
\end{equation}
The corresponding action takes the following form:
\begin{multline} \label{gen-YM-act}
    S[G,A] = \tr \int_X ( (G^{ab})^+ + \gamma (G^{ab})^-)(F_{ab}(A)) - \\
    - \frac{\alpha}{2} (G^{ab})^+ G^+_{ab} -
        \frac{1}{2} (\beta \gamma) G^-_{ab} (G^{ab} )^-\,.
\end{multline}
Varying with respect to $G^{ab}$ one gets: 
\begin{equation} 
\begin{gathered}
    F^+_{ab}(A)  - \alpha G^+_{ab}  = 0 \\
    \gamma(F^-_{ab}(A)  - \beta G^-_{ab})  = 0 
\end{gathered}
\end{equation}

Let us see what does the above system describe at different values of parameters $\alpha, \beta, \gamma$:

1. $\gamma=1$, $\alpha = \beta$. The induced action is that of the  topological YM \eqref{tYM-act}.

2. $\gamma=-1$, $\alpha = \beta$. This reproduces the conventional weak presymplectic gPDE formulation~\cite{Dneprov:2022jyn,Grigoriev:2022zlq}, giving the usual action \eqref{YM-act}. In this case the symplectic structure can be written as~$\omega=\tr \, d(\epsilon_{abcd}G^{cd}\theta^a\theta^b)dC$.

3. $\alpha,\beta,\gamma$ generic. One can express $G_{ab}$ in terms of $F_{ab}(A)$. Then inserting $G_{ab}$ back into the action gives a sum of YM \eqref{YM-act} and topological YM \eqref{tYM-act} actions with some coefficients depending on $\alpha, \beta, \gamma$. 

4. $\gamma=0$, $\alpha,\beta$ -- generic, gives chiral formulation of YM theory, i.e. action~\eqref{F^+^2} up to a numerical factor. Note that in this case $\dl{G_{ab}^-}$ is in the kernel and can be set to $0$, resulting in the subbundle of $E$. In particular, coefficient $\beta$ is irreleveant.

5. $\gamma=0$, $\alpha = 0$ results in the Chalmers-Siegel action functional:
\begin{equation} \label{C-S-action}
    S_{sdYM}[A,G^+] = \int \tr ( G^{+} \wedge (dA + \half[A,A]))\,,
\end{equation}

The analysis of the BV spectrum of these theories with $\gamma \neq 0$ is pretty much the same as in \cite{Dneprov:2022jyn} whilst in the case $\gamma = 0$ the $G^-_{ab}$ component (together with its antifields) lies in the kernel of the presymplectic structure on the super jet-bundle, resulting in the correct BV field-antifield spectrum. 

Note that in the case $\gamma = 0$ the kernel distribution of $\omega$ includes a regular subdistribution generated by $\ddl{}{G^-_{ab}}$ so that $G^-_{ab}$ can be factored out immediately on the initial bundle $E$. This results  in another weak gPDE where the coordinate $G^-_{ab}$ is absent. Finally, the the BV action is given by 
\begin{multline}
    S_{BV,sdYM} = S_{sdYM} + \\
    \int d^4 x \tr(A^*_d(\partial^d C+ [C,A^d]) + \frac{1}{2} C^*  [C,C] + (G^+)^*_{ab}[(G^+)^{ab},C])\,,
\end{multline}
where ghosts and antifields arise as the appropriate components of a supersection.

\subsection{Holst gravity}

The next example is the family of gravity theories, known as Holst gravity~\cite{Holst:1995pc}, which contains the usual Cartan-Weyl Lagrangian as well as its chiral version~\cite{Samuel:1987td,Jacobson:1987yw}.

We immediately start with the associated weak presymplectic gPDE which we take to be a trivial bundle $E\rightarrow T[1]X$, with the fiber being a linear manifold with the following coordinates:
\begin{equation*}
    \begin{gathered}
        \xi^a, \rho_{ab}, \qquad \quad \gh{\xi^a}=\gh{\rho_{ab}}=1\,.      
    \end{gathered}
\end{equation*}
The base $T[1]X$ is coordinatized, as usual, by $x^\mu$, $\gh{x^\mu}=0$ and 
$\theta^\mu$, $\gh{\theta^\mu}=1$. The action of $Q$ is determined by:
\begin{equation}
\begin{gathered}
    Q \xi^a = -\rho^a{ }_{k}\xi^k ,  \\
    Q \rho_{ab} = -\rho_{ak}\rho^k{ }_{b} - \frac{\Lambda}{2}\xi_a \xi_b  , \\
    Q x^{\mu} = \theta^{\mu} ,
\end{gathered}
\end{equation}
where $\Lambda\in \fR$ is a multiple of the cosmological constant.
As the presymplectic structure we take:
\begin{equation}
\label{holst-psymp}
    \omega = \half \epsilon_{abcd}d(\rho^{ab})d(\xi^c \xi^d) + \alpha d(\rho_{ab})d(\xi^a \xi^b),
\end{equation}
where $\alpha \in \fR$ is a parameter. Note that for $\alpha=\pm 1$ the distribution generated by $\dl{\rho^{\mp}_{ab}}$ belongs to the kernel of $\omega$. For $\alpha$ generic the kernel distribution does not have nowhere vanishing vector fields. 

It is easy to check that all the conditions are satisfied. In particular, 
\begin{equation}
    i_Q \omega =  -\half \epsilon^{abcd}d(\xi_a \xi_b \rho_c{ }^k\rho_{kd} + \frac{\Lambda}{4} \xi_a \xi_b \xi_c \xi_d) - \alpha d(\xi^c \xi^d \rho_c{ }^k\rho_{kd})\,.
\end{equation}
If we  parametrize the space of sections by 
\begin{equation}
\begin{gathered}
    \sigma^*(\xi^a) = e^a{ }_{\mu}\theta^\mu\,, \\
    \sigma^*(\rho_{ab}) = \rho_{ab,\mu}\theta^\mu\,, \\      
\end{gathered}
\end{equation}
the associated action takes the following form:
\begin{equation}
    S[e^a, \rho_{ab}] = \int e_a e_b T^{abcd}(d\rho_{cd} + \rho_{c}{ }^k \rho_{kd} + \frac{\Lambda}{4}e_c e_d)\,,
\end{equation}
where $T^{abcd} = \half \epsilon^{abcd} + \frac{\alpha}{2}(\eta^{ac}\eta^{bd}-\eta^{ad}\eta^{bc})$. This is the Holst generalization of the usual Cartan-Weyl action which is reproduced at $\alpha=0$. For generic $\alpha$ the action is still equivalent to that with $\alpha=0$ and hence also describes the standard Einstein gravity. The construction of the BV formulation on the symplectic quotient of the super-jet bundle is a direct generalization of that from~\cite{Grigoriev:2020xec}.

We have seen that the case where $\alpha=\pm 1$ is a special one because the kernel distribution of 
\eqref{holst-psymp} involves (anti)-selfdual components $\rho^{\mp}_{ab}$ of $\rho_{ab}$. In this case (for definiteness we stick to $\alpha=1$) the action takes the form \cite{Samuel:1987td,Jacobson:1987yw,Holst:1995pc}:
\begin{equation}
\label{cPCW}
    S_{cCW}[e^a, \rho_{ab}] = \int e^a e^b R^+_{ab}(\rho)
\end{equation}
where $R^+$ denotes a selfdual part of the curvature. In particular, 
$S_{cCW}$ does not depend on the anti-selfdual component $\omega^-$ of the spin-connection $\omega$. Moreover, the respective BV action does not involve ghosts and antifields associated to $\rho^{-}$ because they belong to the kernel of the presymplectic structure and hence do not enter the induced BV formulation. Nevertheless, the action is still equivalent to the full Einstein gravity. Varying w.r.t. $\omega^+$ gives that $\omega_{ab}^+ = \omega_{ab}^+(e)$ is the selfdual part of the Levi-Chevita connection associated to $e^a$.  Varying w.r.t. $e^a$ tells that the selfdual part of the Riemann tensor associated to $e^a$ is Einstein. This in turn implies that the full Riemann tensor is Einstein.

The formulation in the case of $\alpha=\pm1$ can be achieved by taking a quotient of the initial $E$ by the subdistribution of the kernel distribution, generated by $\dl{\rho^\mp_{ab}}$. In fact this is a maximal regular subbundle of $TE$ that belongs to the kernel distribution. It is easy to check that such a quotient is again a weak presymplectic gPDE and the gauge theory it encodes is precisely~\eqref{cPCW}. 

\subsection{Plebansky gravity}

Consider a trivial bundle $E \rightarrow T[1]X$ with the following coordinates:
\begin{equation}
\begin{gathered}
    B^i \quad gh = 2 \, ,\\
    C^i, \theta^a \quad gh =1 \, , \\
    x^a \quad gh = 0 \, .
\end{gathered}
\end{equation}
Here the index $i$ is the $\mathfrak{su}(2)$ algebra index.
The coordinates $B^i$ are not totally independent, they are subject to the following relation: 
\begin{equation} \label{metricity_cond}
    B^i B^j = \frac{1}{3}\delta^{ij} B^k B^k \, .
\end{equation}
This means that the actual fiber of our system is the nonregular "surface" \eqref{metricity_cond} in the space with coordinates $B^i,C^i$. However, we are going to see that the prolongation of this surface is regular provided one restricts to configurations where the 2-form field associated to $B^i_{ab}$ is in a certain sense nondegenerate.

The action of $Q$ is determined by:
\begin{equation} \label{Pleb-complex}
\begin{gathered}
    Q C^i = -\frac{1}{2} \epsilon^{ijk} C^j C^k + \Lambda B^i \, , \\
    Q B^i = -\epsilon^{ijk}C^j B^k \,,
\end{gathered}
\end{equation}
where $\Lambda \in \fR$ is a numerical parameter. It is easy to see that $Q$ is compatible with the constraints \eqref{metricity_cond}, i.e. $Q$ preserves the ideal of functions on $E$ generated by \eqref{metricity_cond}.

The presymplectic structure is taken to be:
\begin{equation}
    \omega = dC^i dB^i\,, \qquad \gh{\omega}=3\,.
\end{equation}
This structure is nondegenerate and would just define a topological BF theory if we forget about the constraints \eqref{metricity_cond}. However, it has zero vectors on the surface, e.g. vector fields
\begin{equation}
    X^{jk} = (\frac{2}{3}B^i \delta^{jk} - B^j \delta^{ik} - B^k \delta^{ij}) \ddl{}{C^i}
\end{equation}
preserve the surface and satisfy
\begin{equation}
    i_{X^{jk}} \omega = \frac{2}{3} \delta^{jk} B^i dB^i - B^j dB^k - B^k dB^j = d(\frac{1}{3} \delta^{jk} B^i B^i - B^j B^k ) = 0
\end{equation}
modulo terms proportional to~\eqref{metricity_cond}. The presymplectic structure is $Q$ invariant:
\begin{equation}
    i_Q \omega = d(-\half \epsilon^{ijk}C^i C^j B^k + \frac{\Lambda}{2}B^i B^i)\,.
\end{equation}

Let us now turn to the gauge theory our system defines. To this end we parametrize the space of sections as:
\begin{equation}
\begin{gathered}
    \sigma^*(C^i) = \rho^i_{a}\theta^a \,, \\
    \sigma^*(B^i) = \frac{1}{2} B^i_{ab}\theta^a \theta^b \,,
\end{gathered}
\end{equation}
where $ B^i_{ab}(x)$ is subject to the prolongation of the equation~~\eqref{metricity_cond}. The action reads as:
\begin{equation}
    S[\rho, B] = \int B^i(d\rho^i +  \half \epsilon^{ijk}\rho^j \rho^k) - \frac{\Lambda}{2}B^i B^i
\end{equation}
which together with the constraint~\eqref{metricity_cond} gives the Plebansky formulation of gravity~\cite{Plebanski:1977zz,Capovilla:1991qb}. More precisely, the conformal Urbantke
metric can be defined in terms of $B^i_{ab}$ as (see e.g.~\cite{Capovilla:1991qb,Freidel:2012np} for more details)
\begin{equation} \label{Urbantke}
    g_{ab} \sqrt{|g|} = \epsilon^{cdnm}\epsilon^{ijk}B^i_{ac}B^j_{bd}B^k_{nm}\,,
\end{equation}
where it is assumed that $B^i_{ab}$ is subject to the following nondegeneracy conditions: $K^{ij} = B^i_{ab} B^j_{cd} \epsilon^{abcd}$ is invertible. This, in turn, ensures that $g_{ab}$ is invertible. \footnote{If one opts to single out a definite signature for this metric some more conditions are to be imposed on $B^i_{ab}$.}

Let us finally discuss the BV formulation encoded in the above weak presymplectic gPDE with constraints~\eqref{metricity_cond}. To this end consider the super-jet bundle $SJ^\infty(E)$ where we need to take into account the prolongation of the constraints~\eqref{metricity_cond}. More precisely, we again employ the following special coordinates for the $\theta$ jets (see Appendix~\bref{app:jets} for further details) 
\begin{equation}
\begin{gathered}
    (\pi^\infty_E)^*(C^i) = C^i(\theta) = \st{0}{\rho}{ }^i + \st{}{\rho}{ }^i_{a}\theta^a + \frac{1}{2} \st{2}{\rho}{ }^i_{ab}\theta^a \theta^b +  \frac{1}{6}\st{3}{\rho}{ }^i_{abc}\theta^a \theta^b \theta^c + ...\\ 
    (\pi^\infty_E)^*(B^i) = B^i(\theta) = \st{0}{B}{ }^i + \st{1}{B}{ }^i_{a}\theta^a + \frac{1}{2} \st{}{B}{ }^i_{ab}\theta^a \theta^b + \frac{1}{6}\st{3}{B}{ }^i_{abc}\theta^a \theta^b \theta^c + ...
\end{gathered}
\end{equation}
In these coordinates the prolongation of~\eqref{metricity_cond}
are simply the homogeneous in $\theta$ components of
\begin{equation} \label{theta_exp_metr}
    B^i(\theta)B^j(\theta) - \frac{\delta^{ij}}{3}B^k(\theta) B^k(\theta) = 0 \,.
\end{equation}

To see that the local BV formulation encoded in the above system is proper, i.e. takes into account all the gauge symmetries, it is enough to show this for a theory linearized around a vacuum solution $B^i(\theta) = \Sigma^i_{ab}\theta^a \theta^b$, see e.g.~\cite{Grigoriev:2020xec} for a similar analysis.  The linearization of~\eqref{metricity_cond}
reads as 
\begin{equation} \label{exp_metr_cond}
    \begin{gathered}
        \st{0}{b}{ }^i \Sigma^j_{[ab]} + \Sigma^i_{[ab]} \st{0}{b}{ }^j - \frac{2\delta^{ij}}{3} \st{0}{b}{ }^k \Sigma^k_{[ab]} = 0 \,, \\
        \st{1}{b}{ }^i_{[a} \Sigma^j_{bc]} + \Sigma^i_{[ab} \st{1}{b}{ }^j_{c]} - \frac{2\delta^{ij}}{3} \st{1}{b}{ }^l_{[a} \Sigma^l_{bc]} = 0\,, \\
    \end{gathered}
\end{equation}
where $b^i_{...}$ denote the perturbation of the respective $B^i_{...}$ and we only listed explicitly the equations involving $\st{0}{b}{ }^i$ and 
$\st{1}{b}{ }^i$. 

The first equation implies $\st{0}{b}{ }^i = 0$ so there are no ghost $2$ variables.  To analyze the second equation let us take $\Sigma^i_{ab}$ describing a constant frame-field. The simplest choice is to take constant $\Sigma^i_{ab}$ such that $\Sigma^i_{ab}\theta^a \theta^b$ are basis elements in the space of constant selfdual 2-forms on $X$. By employing the  language of two-component spinors one then finds that only 4 independent components survive in $b^i_a$. These can be parameterized in terms of independent $\tilde b_a$  as $\st{1}{b}{ }^i_a =  \epsilon^{abcd} \Tilde{b}{ }_b \Sigma^i_{cd}$. At the same time $\theta$-jets of $C^i$ remain independent. 

All in all we conclude that the BV system associated to this weak presymplectic gPDE model contains all the right fields: 4 diffeomorphism ghosts parameterized by 4 independent components of $\st{1}{b}{ }^i_a$, the selfdual Lorentz ghost $\rho^i$, the frame field inside $B^i_{ab}$ and the Lorentz connection $\rho^i_{a}$. Of course the respective BV antifileds must also be present because the hidden BV system has a nondegenerate symplectic structure.

\subsection{Conformal gravity}

In this section we present a formulation for conformal gravity using the weak presymplectic formalism. This is the refined version of the formulation presented in \cite{Dneprov:2022jyn}. However, the advantage of  the weak presymplectic gauge PDE formalism is that one only works with finite-dimensional objects.

Let $E \rightarrow T[1]X$ be a trivial bundle with the following coordinates:
\begin{equation}
    \begin{gathered}
        \xi^a ,\rho_{ab}, \kappa_a, \lambda , \quad gh = 1 \\
        W^a{ }_{bcd}, C_{abc}, \quad gh = 0
    \end{gathered}
\end{equation}
where $W^a{ }_{bcd}$ has the Weyl ("traceless window") tensor index symmetry and $C_{abc}$ has the Cotton ("traceless hook") tensor index symmetry and $x^\mu,\theta^\mu$ are coordinates on the base $T[1]X$. The action of $Q$ is given by: 
\begin{equation} 
\begin{gathered}
    Q\xi^a = \rho^{a}{ }_c\xi^c + \xi^a \lambda\,, \\
    Q \rho^{a}{ }_b = \rho^{a}{ }_c \rho^{c}{ }_b + (\xi^a\kappa_b - \xi_b\kappa^a) + \frac{1}{2}\xi^c \xi^d W^a{ }_{bcd}\,, \\
    Q \kappa_b = \kappa_c\rho^{c}{ }_b + \lambda \kappa_b + \frac{1}{2}\xi^c \xi^d C_{bcd}\,, \\
    Q\lambda = \kappa_c\xi^c \,,
\end{gathered}
\end{equation}
and
\begin{equation}
\begin{aligned}
    Q W^a{ }_{bcd} =&~P_W(\xi^a C_{bcd}) - \rho_k{ }^a W^k{ }_{bcd}+ \\&~\rho_b{ }^k W^a{ }_{kcd} + \rho_c{ }^k W^a{ }_{bkd} + \rho_d{ }^k W^a{ }_{bck} + 2\lambda W^a{ }_{bcd}\,,\\
    Q C_{abc} =&~ \rho_a{  }^k C_{kbc}+ 
    \\&~\rho_b{  }^k C_{akc} + \rho_c{  }^k C_{abk} + 3\lambda C_{abk}  + \kappa_k W^k{ }_{abc} \,.
\end{aligned}
\end{equation}
Here $P_W$ denotes a projector on the space of rank-4 tensors that singles out a Lorentz-irreducible component of tensor structure of a Weyl tensor.

The presymplectic structure is:
\begin{equation}
\begin{gathered}
    \omega = \omega_W-2\omega_C\,, \\ \omega_W=d(\rho_{ab})d(W^{abnm}\epsilon_{nmpk}\xi^p\xi^k)\,,\qquad \omega_C=d(\xi_a)d(C^a{ }_{bc}\epsilon^{bcpk}\xi_p \xi_k)\,. 
\end{gathered}
\end{equation}
Let us check that the axioms are satisfied. We have:
\begin{multline} \label{iqomw}
    i_{\tgamma}\omega_W = d(\rho_{al}\rho^l{ }_b W^{abnm}\epsilon_{nmpk}\xi^p\xi^k) + (\xi_a\kappa_b - \xi_b\kappa_a)d(W^{abnm}\epsilon_{nmpk}\xi^p\xi^k) + \\ + \frac{1}{2}W_{abij}\xi^i\xi^j d(W^{abnm}\epsilon_{nmpk}\xi^p\xi^k) + d(\rho_{ab})\xi^j P_W(\delta^a_jC^{bnm})\epsilon_{nmpk}\xi^p\xi^k\,,
\end{multline}
and
\begin{equation} \label{iqomc}
\begin{gathered}
    i_{\tgamma} \omega_C = \rho_{an}\xi^n d(C^a{ }_{bc}\epsilon^{bcpk}\xi_p \xi_k) +  d(\xi^a)\rho_a{  }^n C_{nbc} \epsilon^{bcpk}\xi_p \xi_k + \\ + \xi_a \lambda d(C^a{ }_{bc}\epsilon^{bcpk}\xi_p \xi_k) + d(\xi_a) \lambda C^a{ }_{bc}\epsilon^{bcpk}\xi_p \xi_k + \\ + d(\xi^a) \kappa_k W^k{ }_{abc} \epsilon^{bcpk}\xi_p \xi_k\,.
\end{gathered}
\end{equation}
Together they satisfy  $i_Q\omega +d\cL=0$ with:
\begin{multline}
\cL= -\rho_{al}\rho^l{ }_b W^{abnm}\epsilon_{nmpk}\xi^p\xi^k- 2\xi_a\kappa_b W^{abnm}\epsilon_{nmpk}\xi^p\xi^k +2\rho_{an}\xi^n C^a{ }_{bc}\epsilon^{bcpk}\xi_p \xi_k +\\+2\xi_a \lambda C^a{ }_{bc}\epsilon^{bcpk}\xi_p \xi_k - \frac{1}{4}W_{abij}\xi^i\xi^j W^{abnm}\epsilon_{nmpk}\xi^p\xi^k\,. 
\end{multline}
Finally, it is straightforward to see that that $Q\cL = 0$ as well as $i_Qi_Q \omega=0$.

The action functional determined by the above data coincides with the one from \cite{Dneprov:2022jyn}, where it was shown to determine an equivalent formulation of conformal gravity.

\section*{Acknowledgments}
\label{sec:Aknowledgements}

\noindent
We appreciate discussions with J.~Frias, I.~Krasil'shchik, D.~Rudinsky, K.~Druzhkov, A.~Verbovetsky. MG also wishes to thank A.~Chatzistavrakidis, C.~Hull,  K.~Krasnov, A.~Kotov, A.~Sharapov for useful exchanges.
Part of this work was done when ID and MG participated in the thematic program "Emergent Geometries from Strings and Quantum Fields" at the Galileo Galilei Institute for Theoretical Physics, Florence, Italy. The work of VG was supported by Theoretical Physics and Mathematics Advancement Foundation BASIS.

\appendix

\section{Coordinates on superjet space}  \label{app:jets}

Given an adopted coordinate system $\{x^a, \theta^a, u^A\}$ on $E\rightarrow T[1]X$, there exists the associated coordinate system $\{x^a, \theta^a, u^A_{...|...} \}$  on $SJ^\infty(E)$ defined as
\begin{equation}
u^A_{a_1...|b_1...}=D_{a_1}...D^\theta_{b_1}...u^A\,,
\end{equation}
where $D_a = \textrm{tot}(\ddl{}{x^a}), \quad D^\theta_b = \textrm{tot}(\ddl{}{\theta^b})$ are the total derivatives.

However, it appears very useful to employ another natural coordinate system on $SJ^\infty(E)$. The new coordinates are the same $x^a, \theta^a$ and $\psi^A{ }_{a_1...|b_1...}$  defined as: 
\begin{equation} \label{SJ-coord-syst}
\begin{gathered}
  (\psi^A{ }_{a_1...|b_1...}-u^A_{a_1...|b_1...})|_{\theta=0}=0\,, \\
  D^\theta_a \psi^A{ }_{a_1...|b_1...} = 0\,.
\end{gathered}
\end{equation}
Note that in this coordinate system $D^\theta_a = \ddl{}{\theta^a}$.  In particular, coefficients of any evolutionary vector field $s$ 
do not depend on $\theta$ thanks to $\commut{D^\theta_a}{s}=0$.

It is easy to write down the expression of the old coordinates in terms of the new ones:
\begin{equation} \label{SJ-coord-explicit}
\begin{gathered}
    u^A = \psi^A + \theta^a \psi^A_{|a} + \frac{1}{2}\theta^a \theta^b \psi^A_{|ab} + \ldots \\
    u^A_{|a} = \psi^A_{|a} + \theta^b \psi^A_{|ab} + \frac{1}{2}\theta^b \theta^c \psi^A_{|abc}\ldots \\
    \ldots\,.
\end{gathered}
\end{equation}
It is clear that coordinates $\psi^A_{\ldots|\ldots}$ are the ones which are constant along the $\theta$-part of the Cartan distribution. Note that this is a special feature of Grassmann odd base-space coordinates because the even analog, i.e. coordinates $\phi^A_{\ldots|\ldots}$ such that $D_a \phi^A_{\ldots|\ldots}=0$
simply do not exist. Indeed, it is easy to see that $D_a \phi^A_{\ldots|\ldots}=0$
do not have solutions in the space of local functions.

\setlength{\itemsep}{0em}
\small
\providecommand{\href}[2]{#2}\begingroup\raggedright\endgroup

\end{document}